 \documentclass[a4paper,11pt]{scrartcl}
\usepackage{microtype}
\usepackage[utf8x]{inputenc}

\usepackage{tikz}
\usepackage{amssymb}
\usepackage{amsmath}
\usepackage{mathrsfs}
\usepackage{xcolor}

\newtheorem{theorem}{Theorem}[section]
\newtheorem{proposition}[theorem]{Proposition}
\newtheorem{observation}[theorem]{Observation}

\newtheorem{definition}[theorem]{Definition}
\newtheorem{corollary}[theorem]{Corollary}
\newtheorem{lemma}[theorem]{Lemma}
\newtheorem{example}[theorem]{Example}
\newtheorem{remark}[theorem]{Remark}

\newenvironment{proof}{\noindent \textit{Proof. }}{\hfill $\blacktriangleleft$\medskip}
\newenvironment{proofof}[1]{\noindent\textit{Proof of {#1}. }}{\hfill $\blacktriangleleft$\medskip}

\newtheorem{clm}{Claim}

\newcommand{\raus}[1]{}
\def\tu#1{\mathbf{#1}}

\newenvironment{frcseries}{\fontfamily{frc}\selectfont}{}

\newcommand{\calG}{{\mathscr{G}}}

\newcommand{\calA}{{\mathcal{A}}}

\newcommand{\calC}{{\mathcal{C}}}
\newcommand{\calH}{{\mathcal{H}}}

\newcommand{\calP}{{\mathcal{P}}}

\newcommand{\calR}{{\mathcal{R}}}

\newcommand{\calT}{{\mathcal{T}}}

\newcommand{\fhw}{\mathrm{fhw}}

\newcommand{\var}[1]{\textsf{var}({#1})}
\newcommand{\free}[1]{\textsf{free}({#1})}
\newcommand{\atom}[1]{\textsf{atom}({#1})}

\newcommand{\sP}{\mathbf{\#P}}

\newcommand{\NP}{\mathbf{NP}}
\newcommand{\sW}[1]{\mathbf{\#W[#1]}}
\newcommand{\W}[1]{\mathbf{W[#1]}}
\newcommand{\FPT}{\mathbf{FPT}}

\newcommand{\CQ}{\mathrm{CQ}}

\newcommand{\ACQ}{\mathrm{ACQ}}

\newcommand{\sCQ}{\mathrm{\#CQ}}

\newcommand{\classe}[1]{\mathbf{c}(#1)}

\newcommand{\comment}[2]{}

\newcommand{\stefan}[1]{\comment{\color{red}Stefan}{\color{red}#1}}

\PrerenderUnicode{é}

\begin{document}

\title{Structural Tractability of Counting of Solutions to Conjunctive Queries}


 \author{Arnaud Durand\thanks{Partially supported by ANR-11-IS02-0003, project ALCOCLAN}\\ 
 IMJ UMR 7586  -  Logique\\
 Université Paris Diderot\\
 F-75205 Paris, France  \\
  {\small \texttt{durand@math.univ-paris-diderot.fr}}
 \and
 Stefan Mengel\thanks{Partially supported by DFG grants BU 1371/2-2 and BU 1371/3-1.  Furthermore, the research leading to these results has received funding from the [European Community's] Seventh Framework Programme [FP7/2007-2013] under grant agreement n° 238381}\\Institute of Mathematics\\ University of Paderborn\\ D-33098 Paderborn, Germany\\ {\small\texttt{smengel@mail.uni-paderborn.de}} 
 }



\maketitle

\begin{abstract}
In this paper we explore the problem of counting solutions to conjunctive queries.
We consider a parameter called the \emph{quantified star size} of a formula~$\varphi$ which measures how the free variables are spread in $\varphi$. We show that for conjunctive queries that admit nice decomposition properties (such as being of bounded treewidth or generalized hypertree width) bounded quantified star size exactly characterizes the classes of queries for which counting the number of solutions is tractable. This also allows us to fully characterize the conjunctive queries for which counting the solutions is tractable in the case of bounded arity. 
To illustrate the applicability of our results, we also show that computing the quantified star size of a formula is possible in time $n^{O(k)}$ for queries of generalized hypertree width $k$. Furthermore, quantified star size is even fixed parameter tractable parameterized by some other width measures, while it is $\W{1}$-hard for generalized hypertree width and thus unlikely to be fixed parameter tractable. We finally show how to compute an approximation of quantified star size in polynomial time where the approximation ratio depends on the width of the input.
\end{abstract}

%


\section{Introduction}

Conjunctive queries (CQs) are a fundamental class of logical queries that consist of evaluating an existential conjunctive first-order formula over a finite structure. They admit a number of equivalent formulations for example as select-project-join queries in database theory or as homomorphism problems in constraint satisfaction and thus have been extensively studied in various contexts.
Deciding if a Boolean CQ is true or not on a structure is well known to be $\NP$-complete, so the main interest of study has been to identify tractable subclasses, so-called ``islands of tractability'', where the decision question is tractable, i.e.\ can be solved in polynomial time.

One main direction in finding tractable classes of CQs has been imposing  structural restrictions on the formula of the query -- more exactly on the hypergraph associated to it -- while the database is assumed to be arbitrary. In a seminal paper Yannakakis \cite{Yannakakis81} proved that if the formula is acyclic, then the Boolean CQ question becomes tractable. The main idea behind structural restrictions is to extend this result by generalizing it to ``nearly acyclic'' queries. This has lead to many different decompositions for graphs and hypergraphs and associated width measures (see e.g.\ \cite{GottlobLS00,CJG08,miklos08}). The common approach for these decompositions is to group together vertices or edges (of the graphs or hypergraphs)  into clusters of some fixed constant size and to arrange these clusters into a tree. The resulting width measures are often sought to have two desirable properties:
\begin{itemize}
 \item For every $k$ the class of queries of width $k$ should be tractable, i.e.\ Boolean CQ should be solvable in polynomial time.
 \item Given an instance it should be possible to decide if there is a decomposition of width $k$ and construct one if it exists.
\end{itemize}

While decomposition techniques without the first property do not make any sense in the context of CQs, the second property is sometimes relaxed. For some decomposition techniques one does not actually need the decomposition to solve the Boolean query problem \cite{ChenD05}, a promise of the existence is enough. For other decompositions one only knows approximation algorithms that construct decompositions of width that is near the optimal width, which is enough to guarantee tractability of Boolean CQ~\cite{Marx10,AdlerGG07}.

More recently there has also been interest in enumerating all solutions to CQs and in the corresponding counting question. For enumeration of the query answers it turns out that the picture is less clear than for decision \cite{BDG-07,BulatovDGM12,GrecoS10}. Also the situation for counting is more subtle: For quantifier free queries -- which correspond to queries without projections in the database perspective -- most commonly considered structural restrictions yield tractable counting problems (see, e.g.~\cite{PS-11}). While this is nice it is not fully satisfying, because quantifiers/projections are very natural and essential in database queries. While introducing projections does not make any difference for the complexity of Boolean CQ, the situation for the associated counting problem, denoted $\sCQ$,  is dramatically different. In~\cite{PS-11} it is shown that even one single existentially quantified variable is enough to make counting answers to CQs $\sP$-hard even when the structure of the query is a tree (which implies width $1$ for all commonly considered decomposition techniques). This underlines the gain of expressive power obtained by existential quantification in the context of counting. It also follows that the decomposition techniques used for Boolean CQ are not enough to guarantee tractability for counting.

In a previous paper \cite{ouroldpaper} the authors of this paper have proposed a way out of this dilemma for counting by introducing a parameter called \emph{quantified star size} for acyclic conjunctive queries (ACQs). This parameter measures how the free variables are spread in the formula. 
We represented a query formula $\varphi(\tu x)$ with a list $\tu x$ of free variables, by extending the hypergraph $\calH=(V,E)$ associated to $\varphi(\tu x)$ with a set $S\subseteq V$. 
Then the quantified star size is the size of a maximum independent set consisting of vertices from the set $S$ in some specified subhypergraphs of $\calH$. 
It turns out that this measure  precisely characterizes the tractable subclasses of ACQs. The main result is that (under the widely believed assumption $\FPT \ne \sW{1}$ from  parameterized complexity) solutions to a class of ACQs can be counted in polynomial time if and only if the queries in the class are of bounded quantified star size.

\subsection*{Overview of the results}

\paragraph*{Counting solutions to queries}

In this paper we extend the results of \cite{ouroldpaper} from acyclic queries to commonly considered decomposition techniques. To do so we generalize the notion of quantified star size from acyclic queries to general conjunctive queries. We show that every class of CQs that allows efficient counting must be of bounded quantified star size -- again under the same assumption from parameterized complexity. We then go on showing that for all decomposition techniques for CQs commonly considered in the literature combining them with bounded quantified star size leads to tractable counting problems. The key feature that makes this result work is the organization of atoms into a tree of clusters that is prominent in all decomposition methods for $\CQ$s known so far.
Combining the results above we get an exact characterization of the classes of tractable CQ counting problems for commonly considered decomposition techniques. Let us illustrate these results for the example of generalized hypertree decomposition~\cite{GottlobLS00}, which is one of the most general  decomposition methods and one of the most studied too~\cite{GottlobLS00,GottlobMS09,miklos08}.
We have that, under the assumption that  $\FPT \ne \sW{1}$, for any (recursively enumerable) class $\calC$ of hypergraphs of bounded generalized hypertreewidth the following statements are equivalent:
\begin{itemize}
\item $\sCQ$ for instances in $\calC$ can be solved in polynomial time 
\item $\calC$ is of bounded quantified star size.
\end{itemize}

In our considerations, the arity of atoms of queries is not a priori bounded. In this setting, there is no known \textit{ultimate} measure resulting from a decomposition method that fully characterizes tractability even for Boolean $\CQ$. This explains why our characterizations are stated for \emph{each} decomposition method. For bounded arity however, the situation is  different. It is well known that being of bounded treewidth completely characterizes tractability for decision~\cite{GroheSS01,Grohe07} and counting~\cite{DalmauJ04} for CSP (corresponding to quantifier free conjunctive queries in this setting). Combining~\cite{GroheSS01,Grohe07}  and our results from above we derive a complete characterization of tractability for $\sCQ$ in terms of tree width and quantified star size for the bounded arity case.

Note that our results are for counting with set semantics, i.e.\ we count each solution only once. Counting for bag semantics in which multiple occurences of identical tuples are counted has already been essentially solved in~\cite{PS-11}. 

\paragraph*{Discovering quantified star size} To exploit tractability results of the above kind it is helpful if the membership in a tractable class can be decided efficiently, i.e.\ in our case if computing the quantified star size is also tractable. In the second part of the paper, we turn to these  ``discovery problems'' of determining the quantified star size of queries. 

In \cite{ouroldpaper} it is shown that quantified star size of acyclic CQs can be determined in polynomial time. Since star size is equivalent to  independent sets, we cannot expect this to be true on more general queries anymore. Fortunately, it turns out that 
for queries of  generalized hypertree width $k$, there is a $n^k$ algorithm that computes the quantified star size. We show that this is in a sense optimal, because under the assumption $\FPT\ne \W{1}$ there is no efficient (fixed paramater tractable in $k$) algorithm computing the quantified star size for queries parameterized by generalized hypertree width. 

Still some natural decomposition methods admit fixed parameter discovery algorithms. We prove that this is the case for the class of $\CQ$ having  bounded hingetree width (see~\cite{CJG08}). This result is interesting on his own  from a  hypergraph algorithms perspective. Because of the connection between star size and maximum independent set, it provides a new class of hypergraphs for which computing the maximum independent set is $\FPT$. Note that the preceding hardness result shows that fixed parameterized tractability of this problem is unlikely for other hypergraph decomposition techniques.

We then turn our attention to star size approximation.  We show that there is a polynomial time approximation algorithm with ratio $k$ that given a decomposition of width~$k$ runs in time independent of $k$. 

Summing these results up, quantified star size does not only imply tractable counting if combined with well known decomposition techniques, but in case the decomposition is given or can be efficiently computed (hypertreewidth, hingetree width) or approximated (generalized hypertreewidth), then computing quantified star size is itself tractable.

Finally, we investigate the problem of counting solution and computing quantified star size 
for queries of bounded fractional hypertree width~\cite{GroheMarx06, Marx10}. This decomposition method is of a somewhat different nature than the ones studied before so we treat it individually. We again prove that counting is tractable in this setting and that the discovery problem can be decided in $O(n^{k^{O(1)}})$ i.e.\ with a slightly bigger dependency in $k$ than before. 
 
\section{Preliminaries}

\paragraph*{Conjunctive queries} We assume the reader to be familiar with the basics of (first order) logic  (see ~\cite{Libkin-04}). We assume all formulas to be in prenex form. If $\phi$ is a first order formula, $\var{\phi}$ denotes the set of its variables, $\free{\phi}\subseteq  \var{\phi}$ the set of its free variables and $\atom{\phi}$ the set of its atomic formulas. Let $\tu x=x_1,...,x_k$, we denote $\phi(\tu x)$ the formula with free variables  $\tu x$. If $\phi$ is such that  $\free{\phi}=\var{\phi}$ then $\phi$ is said to be \textit{quantifier-free}.
The \textit{Boolean query problem} $\Phi=(\calA,\phi)$ associated to a formula $\phi(\tu x)$  and a structure $\calA$, asks whether the set
\[\phi(\calA)=\{\tu a : (\calA, \tu a)\models \phi(\tu x)\}\]
\noindent called the \textit{query result} is empty or not. The (general) query  problem consists of computing the set $\phi(\calA)$, while the corresponding counting problem is computing the size of $\phi(\calA)$, denoted by $|\phi(\calA)|$. We call two instances $\Phi= (\calA, \phi), \Phi' = (\calA', \phi')$ solution equivalent, if $\free{\phi} = \free{\phi'}$ and $\phi(\calA) = \phi'(\calA')$. When $\phi$ is a $\{\exists,\wedge\}$-first order formula the boolean query problem is known as the \textit{Conjunctive Query Problem}, $\CQ$ for short.  It is well known that the the Boolean $\CQ$ problem is 
$\NP$-complete.
We denote by $\sCQ$ the associated counting problem: given a query instance $\Phi=(\calA,\phi)$, return $|\phi(\calA)|$.

Any $\tu a\in \phi(\calA)$  will be alternatively  seen as an assignment $\tu a: \free{\phi}\rightarrow D$ or as a tuple of dimension $|\free{\phi}|$. Two assignments $\tu a$ and $\tu a'$ are \textit{compatible} (symbol: $\tu a \sim \tu a'$) if they agree on their common variables.

\begin{definition} Let $\phi(\tu x,\tu y)$, $\psi(\tu y,\tu z)$ be two conjunctive queries with  $\tu x\cap\tu z = \emptyset$  and let $\calA,\calA'$ be two finite structures. The the natural join of $\phi$ and $\psi$ is $\phi(\calA)\bowtie \psi(\calA') := \{(\tu a,\tu b,\tu c): \ (\tu a,\tu b)\in \phi(\calA)  \mbox{ and } (\tu b,\tu c)\in \psi(\calA')\}$
\end{definition}

When $\calA=\calA'$, $\phi(\calR)\bowtie \psi(\calA)$ is simply $[\phi\wedge\psi](\calA)$.

\paragraph*{Query size and Model of computation} The underlying model of computation for our algorithms will be the RAM model with unit costs. We assume the relations of a conjunctive query to be encoded by listing their tuples. For a relation $\calR$ let $\mathrm{arity}(\calR)$ denote the arity of $\calR$ and $|\calR|$ the number of tuples in $\calR$. Then the size of an encoding of $\calR$ is $\|\calR\| := \Theta(\mathrm{arity}(\calR) \cdot |\calR|)$. For a vocabulary $\tau$ let $|\tau|$ be the number of predicate symbols. Finally, let $|D|$ be the size of a domain $D$. Then encoding a structure $\calA$ over the vocabulary $\tau$ with domain $D$ takes space $\|\calA\| := |\tau| + |D| +\sum_{\calR\in \tau} \|\calR^{\calA}\|$.

Furthermore, it takes space $\|\phi\| :=\Theta(\sum_{P\in \atom{\phi}} \mathrm{arity(P)})$ to encode a formula $\phi$. The size of an encoding of a CQ instance $\Phi = (\phi, \calA)$ is then $\|\Phi\| := \|\phi\| + \|\calA\|$.

For a detailed discussion and justification of these conventions see \cite[Section~2.3]{FlumFG02}

\paragraph*{Parameterized complexity}

This section is a very short introduction to some notions from parameterized complexity used in the remainder of this paper (for more details see~\cite{FlumGrohe06}). 

A parameterized decision problem over an alphabet $\Sigma$ is a language $L\subseteq \Sigma^*$ together with a computable parameterization $\kappa : \Sigma^* \rightarrow \mathbb{N}$. The problem $(L, \kappa)$ is said to be fixed parameter tractable, or $(L, \kappa)  \in \FPT$, if there is a computable function $f:\mathbb{N}\rightarrow \mathbb{N}$ such that there is an algorithm that decides for $x\in \Sigma^*$ in time $f(\kappa(x)) |x|^{O(1)}$ if $x$ is in $L$.

Let $(L,\kappa)$ and $(L', \kappa')$ be two parameterized decision problems over the alphabets $\Sigma$ resp.\ $\Pi$. 
A parameterized many-one reduction from $(L, \kappa)$ to $(L', \kappa')$ is a function $r: \Sigma^* \rightarrow \Pi^*$ such that for all $x\in \Sigma^*$:
\begin{itemize}
 \item $x\in L \Leftrightarrow r(x)\in L'$,
 \item $r(x)$ can be computed in time $f(\kappa(x)) |x|^{c}$ for a computable function $f$ and a constant $c$, and 
 \item $\kappa'(r(x)) \le g(\kappa(x))$ for a computable function $g$.
 \end{itemize}

It is easy to see that $\FPT$ is closed under parameterized many-one reductions. 

Let $p$-$\mathrm{Clique}$ be the problem of deciding on an input $(G,k)$ where $G$ is a graph and $k$ and integer, if $G$ has a $k$-clique. Here the parameterization $\kappa$ is simply defined by $\kappa(G,k) := k$. The class $\W{1}$ consists of all parameterized problems that are parameterized many-one reducible to $p$-$\mathrm{Clique}$. A problem $(L,\kappa)$ is called $\W{1}$-hard, if there is a parameterized many-one reduction from $p$-$\mathrm{Clique}$ to $(L,\kappa)$.

It is widely believed that $\FPT \ne \W{1}$ and thus in particular $p$-$\mathrm{Clique}$ and all $\W{1}$-hard problems are not fixed parameter tractable. 

Parameterized counting complexity theory is developed similarly to decision complexity. 
A parameterized counting problem is a function $F: \Sigma^* \times \mathbb{N} \rightarrow \mathbb{N}$, for an alphabet $\Sigma$. Let $(x,k)\in \Sigma^*\times \mathbb{N}$, then we call $x$ the input of $F$ and $k$ the parameter. A parameterized counting problem $F$ is fixed parameter tractable, or $F\in \mathbf{FPT}$, if there is an algorithm computing $F(x,k)$ in time $f(k)\cdot |x|^c$ for a computable function $f:\mathbb{N}\rightarrow \mathbb{N}$ and a constant $c\in \mathbb{N}$.

Let $F:\Sigma^*\times \mathbb{N} \rightarrow \mathbb{N}$ and $G:\Pi^*\times \mathbb{N} \rightarrow \mathbb{N}$ be two parameterized counting problems. A parameterized parsimonious reduction from $F$ to $G$ is an algorithm that computes for every instance $(x,k)$ of $F$ an instance $(y,l)$ of $G$ in time $f(k)\cdot |x|^c$ such that $l\le g(k)$ and $F(x,k) = G(y,l)$ for computable functions $f,g:\mathbb{N}\rightarrow \mathbb{N}$ and a constant $c\in \mathbb{N}$. A parameterized $T$-reduction from $F$ to $G$ is an algorithm with an oracle for $G$ that solves any instance $(x,k)$ of $F$ in time $f(k)\cdot |x|^c$ in such a way that for all oracle queries the instances $(y,l)$ satisfy $l\le g(k)$ for computable functions $f,g$ and a constant $c\in \mathbb{N}$. 

Let $p$-$\mathrm{\#Clique}$ be the problem of counting $k$-cliques in a graph where $k$ is the parameter and the graph is the input. A parameterized problem $F$ is in $\sW{1}$ if there is a parameterized parsimonious reduction from $F$ to $p$-$\mathrm{\#Clique}$. $F$ is $\sW{1}$-hard, if there is a parameterized $T$-reduction from $p$-$\mathrm{\#Clique}$ to $F$. As usual, $F$ is $\sW{1}$-complete if it is in $\sW{1}$ and hard for it, too.

Again, it is widely believed that there are problems in $\sW{1}$ (in particular the complete problems) that are not fixed parameter tractable. Thus, from showing that a problem $F$ is $\sW{1}$-hard it follows that $F$ can be assumed to be not fixed parameter tractable.

\paragraph*{Hypergraph decompositions}\label{sct:hypergraphs}

In this section we present some well known hypergraph decompositions methods. For more details and more decomposition techniques see e.g.\ \cite{CJG08,GottlobLS00,miklos08}.

 A (finite) hypergraph  $\calH$ is a pair $(V,E)$ where $V$ is a finite set and $E\subseteq \calP(V)$. We associate a hypergraph $\calH=(V,E)$ to a formula $\phi$ (the \textit{canonical} structure describing $\phi$) by setting $V:= \var{\phi}$ and $E:= \{\var{a}\mid a \in \atom{\phi} \}$.
 
 \begin{example}\label{ex:formula}
 Consider the formula \begin{eqnarray*}\phi &:=& \exists u_1 \exists u_2 \exists u_3 \exists u_4 \exists u_5 \exists u_6 \exists u_7 \exists u_8  \\ && P_1(v_1, u_1) \land P_2(v_2, u_1, u_2)\land P_3(v_2, v_4, u_2, u_3)\\&&\land P_4(v_3, v_4, v_5, u_3, u_4, u_5)\land P_5(v_4, v_5, v_6, v_8) \\&&\land P_6(v_7, v_8, u_5, u_6)\land P_2(v_6, v_9, u_7)\land P_2(v_8, v_9, u_8)\end{eqnarray*}
 The associated hypergraph is illustrated in Figure \ref{fig:hypergraph}.
 \end{example}

 \begin{figure}[t]
 \begin{center}
 \includegraphics[scale=0.6]{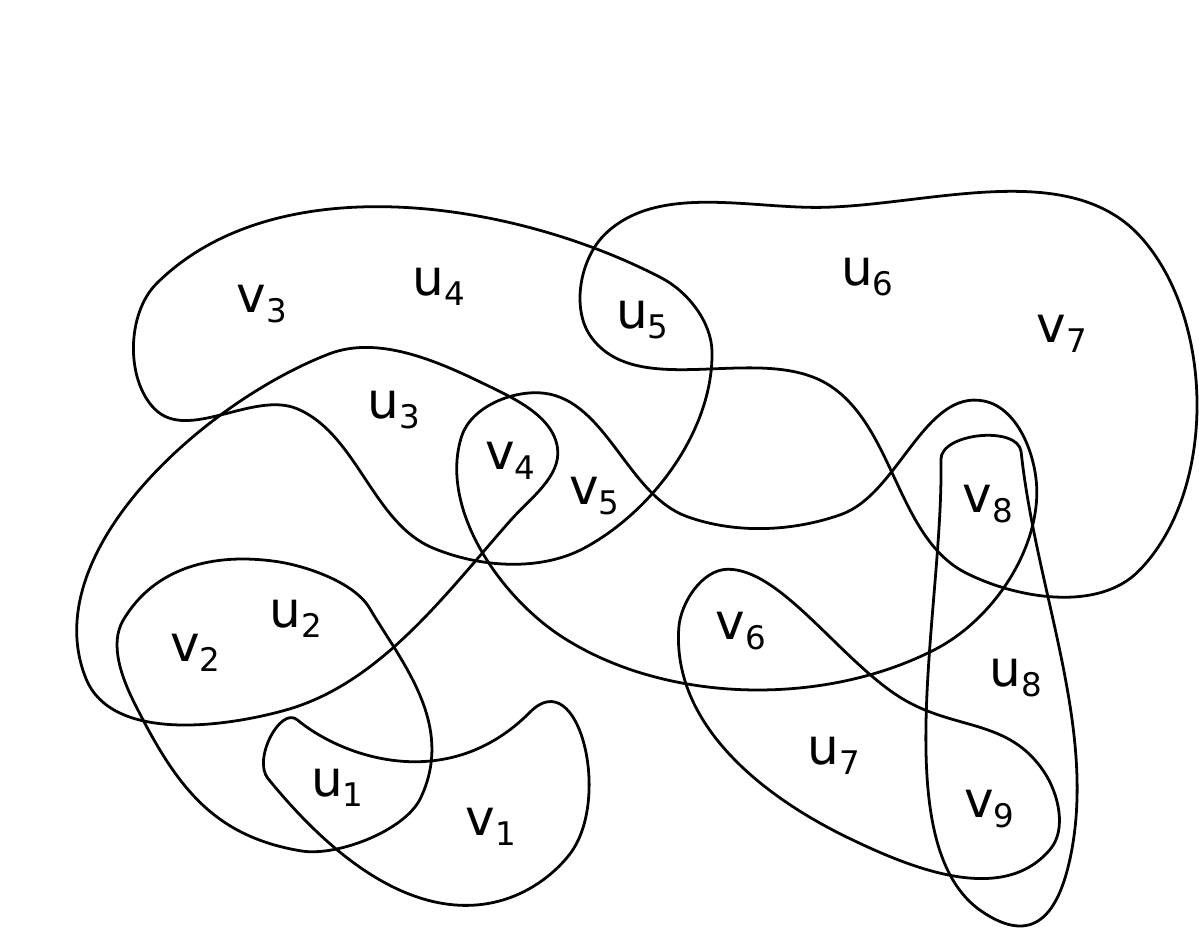}
 \caption{The hypergraph associated to the formula~$\phi$ of Example \ref{ex:formula}.}\label{fig:hypergraph}
 \end{center}
 \end{figure}

An independent set $I$ in $\calH$ is a set of vertices $I\subseteq V$ such that no two of them lie in one edge together. An edge cover $C$ of $\calH$ is an edge set $E' \subseteq E$ such that $\bigcup_{e\in E'} e = V$.
\begin{definition}
 A \emph{generalized hypertree decomposition} of a hypergraph $\calH = (V,E)$ is a triple $(\calT, (\lambda_t)_{t\in T}, (\chi_t)_{t\in T})$ where $\calT = (T,F)$ is a rooted tree and $\lambda_t\subseteq E$ and $\chi_t \subseteq V$ for every $t\in T$ satisfying the following properties:
 \begin{enumerate}
  \item For every $v\in V$ the set $\{t\in T \mid v\in \chi_t\}$ induces a subtree of $\calT$.
  \item For every $e\in E$ there is a $t\in T$ such that $e\subseteq \chi_t$.
  \item For every $t\in T$ we have $\chi_t \subseteq \bigcup_{e\in \lambda_t} e$.
 \end{enumerate}
The first property is called the \emph{connectedness condition}. The sets $\chi_t$ are called \emph{blocks} or \emph{bags} of the decomposition, while the sets $\lambda_t$ are called the \emph{guards} of the decomposition. A pair $(\lambda_t, \chi_t)$ is called \emph{guarded block}.

The \emph{width} of a decomposition $(\calT, (\lambda_t)_{t\in T}, (\chi_t)_{t\in T})$ is defined as $\max_{t\in T}(|\lambda_t|)$. The generalized hypertree width of $\calH$ is the minimum width over all generalized hypertree decompositions of $\calH$.
\end{definition}

We sometimes identify a guarded block $(\lambda_t, \chi_t)$ with the vertex $t$.

\begin{example}
 Figure \ref{fig:decomp} shows a generalized hypertree decomposition of width $3$ for the hypergraph from Figure \ref{fig:hypergraph}.
\end{example}

\begin{figure}[t]
 \begin{center}
 \includegraphics[scale=0.6]{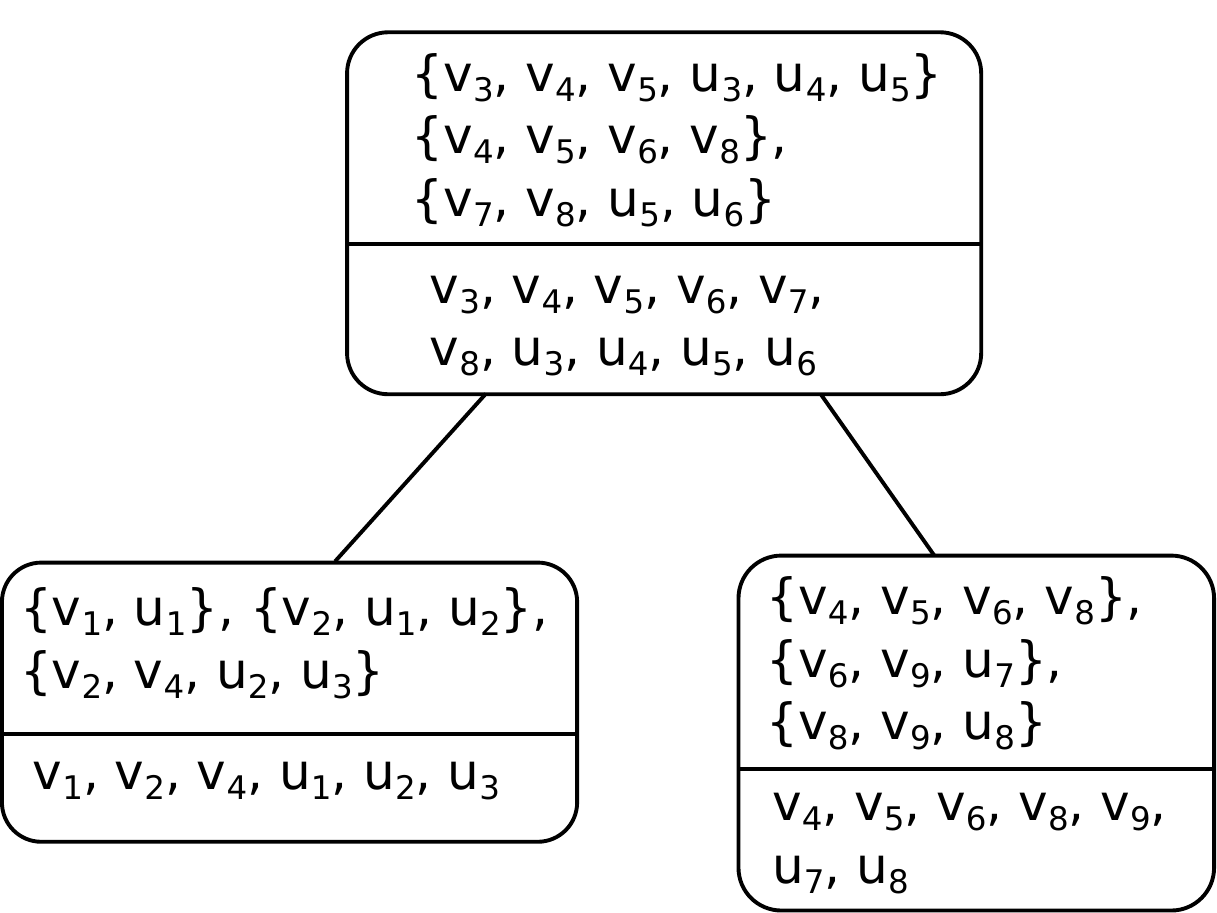}
 \caption{A generalized hypertree decomposition of width $3$ for the hypergraph from Figure \ref{fig:hypergraph}. The boxes are the guarded blocks. In the upper parts the guards are given while the lower parts show the blocks.}\label{fig:decomp}
 \end{center}
 \end{figure}

\begin{definition}
 A hypergraph is acyclic if it has generalized hypertree width $1$.
In this case, the decomposition restricted to its blocks is called a \textit{join tree}.
\end{definition}

Let us fix some notation: For an edge set $\lambda \subseteq E$ we use the shorthand $\bigcup \lambda := \bigcup_{e\in \lambda} e$. For a decomposition $(\calT, (\lambda_t)_{t\in T}, (\chi_t)_{t\in T})$ we write $\calT_t$ for the subtree of $\calT$ that has $t$ as its root. We also write $\chi(\calT_t) := \bigcup_{t'\in V(\calT_t)} \chi_{t'}$.

\begin{definition}
 A generalized hypertree decomposition is called \emph{hingetree decomposition}\footnote{Note that this is not the original definition from \cite{GyssensJC94} but an alternative, equivalent definition from \cite{CJG08}.} if it satisfies the following conditions:
 \begin{enumerate}
  \item[4.] For each pair $t_1,t_2\in T$ with $t_1\ne t_2$ there are edges $e_1 \in \lambda_{t_1}$ and $e_2\in \lambda_{t_2}$ such that $\chi_{t_1} \cap \chi_{t_2} \subseteq e_1\cap e_2$.
  \item[5.] For each $t\in T$ we have $\bigcup \lambda_t = \chi_t$.
  \item[6.] For each $e\in E$ there is a $t\in T$ such that $e\in \lambda_t$.
 \end{enumerate}
\end{definition}

Hingetree width (also called degree of cyclicity) is defined analogously to generalized hypertree width.

\begin{example}
 The decomposition from Figure \ref{fig:decomp} is also a hingetree decomposition.
\end{example}

\begin{definition}
 The \emph{primal graph} of a hypergraph $\calH=(V,E)$ is the graph $\calH_P = (V, E_p)$ with $E_p := \{uv \in \binom{V}{2} \mid \exists e \in E: u,v\in E\}$. 
\end{definition}

\begin{definition}
 A \emph{tree decomposition} of a hypergraph $\calH$ is a generalized hypertree decomposition of its primal graph $\calH_P$. The width of a tree decomposition is the size of its biggest bag minus $1$. The \emph{treewidth} of $\calH$ is the minimum width over all tree decompositions of $\calH$.
\end{definition}

For all decompositions defined above we define the width of a $\CQ$-instance to be the width of the associated hypergraph.

We now recall some known results on the various decomposition methods.
\begin{lemma}
 \begin{enumerate}
  \item[a)] (see e.g.\ \cite{CJG08}) For all of the width measures defined above Boolean $\CQ$-instances of width $k$ can be solved in time $\|\Phi\|^{p(k)}$ for a polynomial $p$.
  \item[b)] (\cite{GyssensJC94}) There is an algorithm that given a hypergraph $\calH=(V,E)$ computes a minimum width hingetree decomposition 
  in time $|V|^{O(1)}$.
  \item[c)] (\cite{Bodlaender93})Computing 
  minimum width tree decompositions is fixed parameter tractable parameterized by the treewidth.
  \item[d)] (\cite{AdlerGG07,GottlobLS02}) There is an algorithm that given a hypergraph $\calH=(V,E)$ of generalized hypertree width $k$ constructs a generalized hypertree decomposition of width~$O(k)$ of~$\calH$ in time $|V|^{O(k)}$. 
 \end{enumerate}
\end{lemma}

\begin{definition}
Let $\calH= (V,E)$ be a hypergraph and $V' \subseteq V$.
The \emph{induced subhypergraph}  $\calH[V']$ of $\calH$ is the hypergraph $\calH[V'] = (V', \{e \cap V' \mid e \in E, e \cap V' \ne \emptyset\})$.

Let $x,y \in V$, a \textit{path} between $x$ and $y$ is a sequence of vertices $x=v_1,...,v_k=y$ such that for each $i\in [k-1]$ there is an edge $e_i \in E$ with $v_i,v_{i+1} \in e_i$. 

A (connected) component of $\calH$ is the induced subhypergraph $\calH[V']$ for a maximal vertex set $V'$ such that for each pair $x,y\in V'$ there is a path between $x$ and $y$ in $\calH$.
\end{definition}

\begin{observation}\label{obs:subhypergraphwidth}
 Let $\beta$ be any decomposition technique defined in this section. Let $\calH=(V,E)$ be a hypergraph of $\beta$-width $k$. Then for every $V'\subseteq V$ the induced subhypergraph $\calH[V']$ has $\beta$-width at most $k$.
\end{observation}
\begin{proof}
 Let $(\calT, (\lambda_t)_{t\in T}, (\chi_t)_{t_in T})$ be a $\beta$-decomposition of $\calH$ of width $k$. For each guarded block $(\lambda_t, \chi_t)$ compute a guarded block $(\lambda'_t, \chi'_t)$ with $\chi_t := \chi_t \cap V'$ and $\lambda_t :=\{e\cap V' \mid e\in \lambda\}$. It is easy to check that $(\calT, (\lambda'_t)_{t\in T}, (\chi'_t)_{t_in T})$ is a $\beta$-decomposition of width at most $k$.
\end{proof}

\section{Quantified-star size}

In this section we generalize quantified star size which was introduced in \cite{ouroldpaper} for acyclic conjunctive queries to general conjunctive queries.

\begin{definition}
Let $\calH=(V,E)$ be a hypergraph and $S\subseteq V$. Let $C$ be the vertex set of a connected component of $\calH[V-S]$.
Let $E_{C}$ be the set of hyperedges $\{e\in E\mid e \cap C \ne \emptyset\}$ and $V':= \bigcup_{e\in E_C} e$. Then $\calH[V']$ is called an \emph{$S$-component} of $\calH$.
\end{definition} 


\begin{definition}Let $\calH=(V,E)$ be a hypergraph.
For a set $S\subseteq V$ the \emph{$S$-star size} of $\calH$ is the maximum size of an independent set consisting only of vertices in $S$ in an $S$-component of $\calH$. We say that this independent set forms the $S$-star.
\end{definition}

\begin{example}
Take $S=\{v_1,...,v_9\}$ in the hypergraph of  Figure~\ref{fig:hypergraph}. It has  three $S$-components with respective  edge lists:
\begin{enumerate}
\item $\{v_1,u_1\}$, $\{v_2,u_1,u_2\}$, $\{v_2,v_4,u_2,u_3\}$, $\{v_7,v_8,u_5,u_6\}$, $\{v_4,v_5,v_6,v_8\}$, $\{v_3,v_4,v_5,u_3,u_4,u_5\}$,      $\{v_8\}$ 
\item  $\{v_8,v_9,u_8\}$,  $\{v_8\}$ 
\item  $\{v_6,v_9,u_7\}$,  $\{v_6\}$ 
\end{enumerate}
The $S$-star size i.e.\ the size of a maximum independent of $S$-vertices in a $S$-component is $4$. The set $\{v_1,v_2,v_3,v_7\}$  forms an $S$-star (there are several other possibilities).
\end{example}

It is easy to see that for acyclic hypergraphs this definition of $S$-star size coincides with the definition in \cite{ouroldpaper} which was only defined for acyclic hypergraphs.

\begin{definition}
An $S$-hypergraph is a pair $(\calH, S)$ where $\calH=(V,E)$ is a hypergraph and $S \subseteq V$. To each formula $\phi$ we associate an $S$-hypergraph $(\calH, S)$ where $\calH$ is the hypergraph associated to $\phi$ and $S:= \free{\phi}$. The \emph{quantified star size} of a $\CQ$ instance $\Phi=(\calA, \phi)$ is the $S$-star size of $(\calH, S)$.
\end{definition}

Let $\calG_{star}$ be the class of $S$-hypergraphs $(\calH_n, S_n)$, $n\in\mathbb{N}$, where $\calH_n$ is a star graph and $S_n$ consists of its leaves. More precisely, $H_n=(V_n,E_n), S_n$ are defined as
\begin{itemize}
\item $V_n=\{z,y_1,....,y_n\}$, 
\item $E_n=\{\{z,y_i\} \mid i=1,...,n\},$ 
\item $S_n=\{y_1,....,y_n\}.$ 
\end{itemize}

We will use the following lemma from \cite{ouroldpaper} to which we give an alternative simpler proof below.

\begin{lemma}[\cite{ouroldpaper}]\label{lem:hardforstarsB}
 $\sCQ$ is $\sW{1}$-hard restricted to instances that have $S$-hypergraphs in $\calG_{star}$ parameterized by the size of the stars.
\end{lemma}

\begin{proof} We show the hardness by a parameterized $T$-reduction from $p$-$\mathrm{\#Clique}$. The basic idea is that instead of counting $k$-cliques in a graph, we can also count the $k$-tuples of vertices that are not a clique. 

So let $G=(V,E)$ be a simple, undirected graph and $k\in \mathbb{N}$. A tuple $(v_1, \ldots, v_k)\in V^k$ is not a clique if and only if it there are $i,j\in [k], i\ne j$ such that $v_i v_j$ is not an edge. Observe that because $G$ is loopless this is necessarily true if $(v_1, \ldots, v_k)$ contains a double vertex. We will show how to check if a tuple $(v_1, \ldots v_k)$ is a clique with a $\CQ$-instance of the prescribed form. 

We construct a $\sCQ$-instance $\Phi= (\calA, \phi)$ with $\phi := \exists z \bigwedge_{i\in [k]} P_i(z, v_i)$. Clearly the formula is of the right form. The domain of $\calA$ is $D:= V\cup (V\times V\times [k]\times [k])$. For each $i\in [k]$ the structure $\calA$ has the relation
 \begin{align*}P_i^\calA :=& \cup \{ ((v,w,i,j),v), ((w,v,j,i),v) \mid v,w\in V \\
 &\hskip 2cm  v\ne w, vw\notin E, j\in [k], j\ne i\}\\
& \cup (V\times V\times ([k]\setminus \{i\}) \times ([k]\setminus \{i\}))\times V.\end{align*}
This completes the construction of $\Phi$.

First, observe that $\Phi$ can be constructed in time polynomial in $|G|$ and $k$, so if we can compute the number of $k$-cliques of $G$ from $|\phi(\calA)|$ sufficiently quickly, the construction is indeed a parameterized $T$-reduction.

Furthermore, observe that for each satisfying assignment the variables $v_1, \ldots, v_k$ take only values in $V$. We claim that an assignment $a: \{v_1, \ldots, v_k\}\rightarrow D$ satisfies $\phi$ if and only if $a(v_1), \ldots , a(v_k)$ is not a clique of size $k$ in $G$. Essentially, the quantified variable $z$ here guesses the edge that is missing between $v_i$ and $v_j$.

Indeed, if $a(v_1), \ldots , a(v_k)$ is a tuple of vertices such that two vertices in it are not adjacent, say $a(v_i)= x_i$, $a(v_j)= x_j$, $x_ix_j\notin E$, then assigning $(x_i,x_j,i,j)$ to $z$ satisfies all atoms.

Let on the other hand $a(v_1), \ldots , a(v_k)$ be a clique of size $k$ in $G$. We claim that there is no assignment to $z$ that satisfies all atoms. Clearly in a satisfying assignment $z$ can take no value in $V$. So $z$ must take a value in $V\times V\times [k]\times [k]$, say $(v,w,i,j)$. But then in particular $P_i(z,v_i)$ and $P_j(z,v_j)$ are satisfied. It follows that $a(v_i) = v$, $a(v_j)= w$, $v,w\notin E$, which is a contradictiton. So indeed,  $a(v_1), \ldots , a(v_k)$ is a clique of size $k$ in $G$ if and only if $a$ is a satisfying assignment.

It follows that the number of cliques in $G$ is $\frac{1}{k!} (|V|^k \setminus |\phi(\calA)|)$. But $|V|^k$ and $k!$ can be easily computed in time $(k|V|)^{O(1)}$ and thus one can compute the number of $k$-cliques of $G$ from $|\phi|$, $G$ and $k$ in time $(k|V||\phi|)^{O(1)}$ which completes the reduction.
\end{proof}

\section{The complexity of counting}\label{sct:countGHW}

In this section we show that the decomposition techniques introduced in Section \ref{sct:hypergraphs} lead to efficient counting when combined with bounded quantified star size. Furthermore, we show that bounded quantified star size is necessary for efficient counting under standard assumptions.

\begin{theorem}\label{thm:guarded}
 There is an algorithm that given a $\sCQ$-instance $\Phi=(\calA,\phi)$ of quantified starsize $\ell$ and a generalized hypertree decomposition $\varXi = (\calT,$ $ (\lambda_t)_{t\in T}, (\chi_t)_{t\in T})$ of $\Phi$ of width $k$ counts the solutions of $\phi$ in time $\|\Phi\|^{p(k,\ell)}$ for a fixed polynomial $p$.
\end{theorem}

In the proof we will use the following lemma from \cite{ouroldpaper}.

\begin{lemma}\label{lem:edgecover}
 For acyclic hypergraphs the size of a maximum independent set and a minimum edge cover coincide. Moreover, there is a polynomial time algorithm that given an acyclic hypergraph $\calH$ computes a maximum independent set~$I$ and a minimum edge cover $E^*$ of $\calH$.
\end{lemma}

\begin{proof}[of Theorem \ref{thm:guarded}]
 Given $\Phi=(\calA,\phi)$, we construct a solution equivalent instance $\Phi''$ in two steps which is of generalized hypertree width~$k$, too, and has a quantifier free formula.
 
 Let $\calH= (V,E)$ be the hypergraph of $\phi$.  Let $V_1, \ldots, V_m$ be the vertex sets of the components of $\calH[V-S]$ and let $V_1', \ldots, V_m'$ be the vertex sets of the $S$-components of~$\calH$. Clearly, $V_i\subseteq V_i'$ and $V_i'-V_i = V_i'\cap S =: S_i$. Let $\Phi_i$ be the $\sCQ$-instance whose formula $\phi_i$ is obtained by restricting all atoms of $\phi$ to the variables in $V_i'$ and whose structure $\calA_i$ is obtained by projecting all relations of $\calA$ accordingly.  The associated hypergraph of $\phi_i$ is $\calH[V_i']$ and $\calH[V_i']$ has a generalized hypertree decomposition~$\varXi_i$ of width at most $k$ with tree a $\calT_i$ that is a subtree of $\calT$ (see Observation~\ref{obs:subhypergraphwidth}). 

For each $\Phi_i$ we construct  a new $\sCQ$-instance $\Phi_i'=(\calA_i',\phi_i')$ as follows. For each guarded block $b = (\lambda, \chi)\in \varXi_i$ we construct a new atomic formula $\varphi_b$ in the variables $\chi$. The associated relation is given by $\pi_{\chi}(\bowtie_{\phi \in \atom{\Phi_i}\colon \var{\phi} \subseteq \bigcup \lambda} \phi)$ i.e.\ by taking the natural join of all relations whose variables are guarded in the guarded block and projecting on~$\chi$. The formula $\phi_i'$ for  $\Phi_i'$ is obtained as the conjunction of all~$\varphi_b$. The decomposition $\varXi_i$ has width at most $k$ so this can be done in time $\|\Phi\|^{O(k)}$. Obviously, $\Phi_i$ and $\Phi'_i$ are solution equivalent. Furthermore $\phi'_i$ is acyclic, because it has a decomposition with tree $\calT_i$, the same blocks as $\varXi_i$ and width $1$. 
 Let $\calH_i$ be the associated hypergraph of $\phi_i'$, then $\calH_i$ has only one single $S_i$-component, because all the vertices in $V_i$ are connected in~$\calH$ and thus also in $\calH_i$. Also the $S_i$-star size of $\calH_i$ is at most $\ell$. To see this consider two independent vertices $u,v$ in~$\calH_i$. The edges of $\calH_i$ are equal to the blocks of $\varXi_i$, so $u$ and~$v$ do not appear in a common block in $\varXi_i$. But then $u$ and $v$ cannot appear in one common block in $\varXi$, because of~$\calT$ being a tree and the connectedness condition. So $u$ and $v$ are independent in $\calH$, too, and thus every independent set in $\calH_i$ is also independent in $\calH$. So $\calH_i$ indeed has $S_i$-star size at most $\ell$. Thus the vertices in $S_i$ can be covered by at most $\ell$ edges $e_1, \ldots , e_\ell$ in $\calH_i$ which we can compute in polynomial time by Lemma \ref{lem:edgecover}. Let $\alpha_1, \ldots , \alpha_\ell$ be the corresponding atoms. We again construct a new atomic formula $\phi_i''$ in the variables $S_i$ only and an associated relation $\calA_i''$ as follows:  For each combination $t_1, \ldots, t_\ell$ of tuples in $\alpha_1(\calA_i'), \ldots, \alpha_{\ell}(\calA_i')$ fix the free variables in $\phi_i'$ to the constants prescribed by the tuples $t_1, \ldots, t_\ell$ if these do not contradict. If the resulting $\CQ$ instance has a solution, add the projection of $t_1\bowtie \ldots \bowtie t_\ell$ on $S_i$  to the relation $\calA_i''$ of $\phi_i''$. By construction $\Phi_i'$ and $(\calA_i'',\varphi_i'')$ are solution equivalent.
 Observe that the instances to be solved in this construction are tractable \cite{Yannakakis81}, so all of this can be done in time $\|\Phi_i\|^{p(k,\ell)}$ for a polynomial $p'$.

We now eliminate all quantified variables in the original formula $\phi$. To do so we add the atom $\phi_i''$ for $i\in [m]$ and delete all atoms that contain any quantified variable, i.e.\ we delete each~$\phi_i'$. Add the $\calA_i''$ to the structure $\calA$ and call the resulting $\sCQ$ instance $\Phi''=(\calA'',\phi'')$. Because $(\calA_i'',\phi_i'')$ is solution equivalent to $\Phi_i'$, we have that $\Phi$ and $\Phi''$ are solution equivalent, too. We construct a guarded decomposition of $\phi''$ by doing the following: For each guarded block $(\lambda, \chi)$ of $\varXi$ with $\chi\cap V_i\ne \emptyset$ we construct a guarded block $(\lambda', \chi')$ by deleting all edges~$e$ with $e\cap V_i\ne \emptyset$ from $\lambda$ and adding  the edge $S_i$ for $\phi_i''$. Furthermore we set $\chi' = (\chi-V_i)\cup S_i$. It is easy to see that the result is indeed a generalized hypertree decomposition of~$\phi''$ of width at most $k$.

With standard techniques (see e.g.\ \cite{CJG08}) we construct in polynomial time a quantifier free acyclic $\sCQ$-instance that is solution equivalent to $\Phi''$. Its solutions and thus those of $\Phi$ can then be counted with the algorithm in \cite{PS-11} or \cite{ouroldpaper}.
\end{proof}

We now show that bounded quantified star size is necessary for efficient counting no matter which other structural restrictions we put on $S$-hypergraphs.

\begin{lemma}\label{lem:tightgeneral}
 Let $\calG$ be a recursively enumerable class of $S$-hypergraphs such that $\sCQ$ for all instances whose $S$-hypergraph is in $\calG$ is fixed parameter tractable parameterized by the size of the formulas. Then $\calG$ has bounded $S$-star size or $\sW{1} = \FPT$.
\end{lemma}
\begin{proof}[sketch]
 The proof is a generalization of the respective proof in \cite{ouroldpaper}:
We show that if the $S$-star size of $\calG$ is not bounded, then there is an FPT algorithm for $\sCQ$ on $\calG_{star}$, the class of stars with a single quantified variable in the center. As this problem is $\sW{1}$-hard by Lemma \ref{lem:hardforstarsB}, it follows that $\sW{1}=\FPT$.

So assume that $\sCQ$ is tractable on $\calG$ and $\calG$ has unbounded $S$-star size. We will construct  a fixed parameter algorithm for $\sCQ$ on $\calG_{star}$. So let $\Phi = (\calA,\varphi)$ be an instance of this latter problem, i.e.\ $\Phi$ has the formula $\varphi:= \exists z \bigwedge_{i=1}^k E_i(y_i, z)$. Let the domain of $\calA$ be $D$. Because~$\calG$ is recursively enumerable and of unbounded $S$-star size, there is a computable function $g:\mathbb{N}\rightarrow \mathbb{N}$ such that for $k\in \mathbb{N}$ one can compute $(\calH, S)\in \calG$ with $S$-star size at least $k$ in time $g(k)$. We will embed $\Phi$ into $\calH$ to construct an $\sCQ$-instance $\Phi' = (\calA',\psi)$ of size $g(k)n^{O(1)}$ where $n$ is the size of $\Phi$. Furthermore, $\psi$ will have the $S$-hypergraph~$\calH$ and $\calA'$ the same domain $D$ as $\calA$. For convenience, $\Phi'$ will be built on a language containing one distinct relation symbol for each hyperedges in $\calH$.

Let $\calH'$ be the $S$-component of $\calH$ that contains $k$ independent vertices in the respective primal component. Call these vertices $s_1, \ldots s_k$. We may assume that the $s_i$ are also computed in time $g(k)$ during the construction of $\calH$. Observe that there must be a vertex $v$ that is connected to each of the $s_i$ by a path $P_i$ such that the only vertex in $P_i$ that is in~$S$ is $s_i$, because all the $s_i$ lie in the same $S$-component. We now construct a $\sCQ$ instance $\Phi'$ that has the associated $S$-hypergraph $\calH$.

All vertices that do not lie on any $P_i$ are forced to a dummy value $d$ in a straightforward way by all their constraints. All vertices on the $P_i$ that are no $s_j$ may take arbitrary but equal values in $D$. This is possible, because they are all connected to the common vertex $v$ by paths. Let $v_i$ be the predecessor of $s_i$ on $P_i$. For all constraints that contain $v_i$ and $s_i$ we allow for them exactly the combinations allowed by the relation of $E_i^\calA$. Observe that there is no edge that contains more than one of the $s_i$ by definition, so each constraint has at most~$|D|^2$ tuples.

Clearly, $\Phi$ and $\Phi'$ have the same number of solutions. Furthermore, we have $|\psi| \le g(k)$ and $\Phi'$ can be constructed in time at most $g(k) \|\Phi\|^2$, because $\calH$ has size at most $g(k)$ and the size of the relations for the constraints is bounded by~$|D|^2$. But by assumption the solutions of $\Phi'$ can be counted in time $h(|\psi|)\|\Phi'\|^c$ for a constant $c$ and a computable function $h$. Thus the solutions of $\Phi$ can be counted in time $h(g(k))\|\Phi\|^{c'}$ for a constant $c'$ which completes the proof.
\end{proof}

With Theorem \ref{thm:guarded} and Lemma \ref{lem:tightgeneral} we have a solid understanding of the complexity of $\sCQ$ for structural classes that can be characterized by restrictions of generalized hypertree width. For each decomposition method with what Cohen et al. \cite{CJG08} call the ``tractable construction'' property, i.e.\ there must be a way to construct a decomposition efficiently, quantified star size is essentially the only parameterization that makes counting tractable. For the definitions of decomposition techniques not defined in this paper see~\cite{GottlobLS00}.

\begin{corollary}
Let $\beta$ be one of the following decomposition techniques:
biconnected component, cycle-cutset,  cycle-hypercutset, hinge-tree, hypertree, or  generalized hypertree decomposition.
Let furthermore $\calG$ be a recursively enumerable class of $S$-hypergraphs of bounded $\beta$-width. Then counting solutions to all $\sCQ$-instances whose associated hypergraph is in  $\calG$ is tractable if and only if $\calC$ is of bounded $S$-star size (assuming $\FPT \ne \sW{1}$).
\end{corollary}

\section{An optimal result for bounded arity}

In this section we show that for bounded arity $\sCQ$ we can exactly characterize which classes of instances allow polynomial time counting. This result is derived by combining the results of the preceding sections and the following theorem 
from~\cite{GroheSS01} that we rephrase in our slighlty different wording.  

\begin{theorem}[\cite{GroheSS01}] \label{thm:grohe}Let $\calG$ be a recursively enumerable class of hypergraphs of bounded arity. Assume $\FPT \ne \W{1}$. Then the following three statements are equivalent:
\begin{itemize}
 \item Boolean $\CQ$ for all instances with hypergraphs in $\calG$ can be decided in polynomial time.
 \item Boolean $\CQ$ for all instances with hypergraphs in $\calG$ is fixed parameter tractable parameterized by the size of the formulas.
 \item The hypergraphs in $\calG$ are of bounded treewidth.
\end{itemize}
\end{theorem}

 Theorem \ref{thm:grohe} is originally stated to be true even for every fixed vocabulary. It
 has been generalized to any recursively enumerable class of conjunctive formulas \cite{Grohe07}. In this context,  a characterization of tractability for counting solutions of \textit{quantifier-free} conjunctive queries is given in \cite{DalmauJ04} in almost the same terms as Theorem \ref{thm:grohe} but with the weaker assumption that $\FPT \ne \sW{1}$. 
 We show here a complete characterization of tractability for counting for general conjunctive queries. Not too surprisingly, tractability depends on both treewidth and star size of the underlying hypergraph.
 
\begin{theorem}\label{thm:bounded}
 Let $\calG$ be a recursively enumerable class of $S$-hypergraphs of bounded arity. Assume that $\W{1} \ne \FPT$. Then the following statements are equivalent:
 \begin{enumerate}
  \item \label{bnd:1} $\sCQ$ for all instances whose $S$-hyper\-graph is in $\calG$ is solvable in polynomial time.
  \item \label{bnd:2} $\sCQ$ for all instances whose $S$-hyper\-graph is in $\calG$ is fixed parameter tractable parameterized by the size of the formulas.
  \item \label{bnd:3} There is a constant $c$ such that for each $S$-hypergraph $(\calH,S)$ in $\calG$ 
  the treewidth of $\calH$ and the $S$-star size
 are at most $c$.
\end{enumerate}
 \end{theorem}
\begin{proof}
The direction \ref{bnd:1} $\rightarrow$ \ref{bnd:2} is trivial. Furthermore, \ref{bnd:3} $\rightarrow$ \ref{bnd:1} follows directly from Theorem \ref{thm:guarded}. So it remains only to show \ref{bnd:2} $\rightarrow$ \ref{bnd:3}.

So assume that there is a recursively enumerable class~$\calG$ of $S$-hypergraphs such that counting solutions to $\sCQ$-instances whose $S$-hypergraph are in $\calG$ is fixed parameter tractable but \ref{bnd:3} is not satisfied by $\calG$. From Lemma \ref{lem:tightgeneral} we know that the $S$-starsize of $\calG$ must be bounded, so it follows that the treewidth of $\calG$ must be unbounded.

We construct a class $\calG'$ of hypergraphs by doing the following: For each $S$-hypergraph $(\calH,S)$ in $\calG$ we add $\calH$ to~$\calG'$. Clearly $\calG'$ is recursively enumerable and of unbounded treewidth. We will show that Boolean $\CQ$ for all instances whose hypergraphs are in $\calG'$ is fixed parameter tractable parameterized by the size of the formula. This leads to a contradiction with Theorem \ref{thm:grohe}.
 
 Because $\calG$ is recursively enumerable, there is an algorithm that that for each $\calH$ in $\calG'$ constructs an $S$-hypergraph $(\calH,S)$ in~$\calG$ that has lead to the addition of $\calH$ to $\calG'$. For example one can simply enumerate all $S$-hypergraphs in $\calG$ until finding such a $(\calH,S)$. Let $f(\calH)$ be the number of steps the algorithm needs on input $\calH$. The function $f(\calH)$ is well defined and computable. We then define $g:\mathbb{N} \rightarrow \mathbb{N}$ by setting $g(k) := \max_\calH(f(\calH))$, where the maximum is over all hypergraphs $\calH$ of size $k$ in $\calG'$. Clearly, $g$ is again well defined and computable. Thus for each $\calH$ in $\calG'$ we can compute in time $g(|\calH|)$ an $S$-hypergraph $(\calH,S)$ in $\calG$.
 
 Now let $\Phi=(\calA,\phi)$ be a $\CQ$-instance with hypergraph $\calH$ in~$\calG'$. To solve it we first compute $(\calH,S)$ as above and construct a $\sCQ$-instance $\Psi = (\calA,\psi)$ with $(\calH,S)$ as associated $S$-hypergraph for $\psi$  by adding existential quantifiers for all variables not in $S$. 
 Obviously $\Phi$ has solutions if and only if $\Psi$ has any. 
  But by assumption the solutions of $\Psi$ can be counted in time $h(|\psi|)\|\Psi\|^{O(1)}$ for some computable function $h$, so $\Phi$ can be decided in time $(g(|\phi|) + h(|\phi|)) \|\Phi\|^{O(1)}$ and thus is fixed parameter tractable. This is the desired contradiction to Theorem \ref{thm:grohe}.
 \end{proof}

\begin{remark} Note that our characterization relies on the underlying hypergraph structures of the query. In~\cite{Grohe07,DalmauJ04}, the corresponding characterizations are stronger in the sense that they are true for any recursively enumerable class of conjunctive formulas. Also these results and the one from \cite{GroheSS01} can be proved for every fixed vocabulary, while our proofs of the Lemmas \ref{lem:hardforstarsB} and \ref{lem:tightgeneral} and thus also Theorem \ref{thm:bounded} rely on the fact that we can choose our vocabulary in the construction. It remains an open question whether our result can be improved similarly to the others.  

Also, the result in~\cite{DalmauJ04} (for quantifier free $\sCQ$) is proved under the weaker assumption $\sW{1} \ne \FPT$. Showing the same equivalent result for general $\sCQ$ seems to be hard since our case also contains decision problems (e.g.\ $\sCQ$ with no free variables).
\end{remark}

\section{Computing star size}

In this section we consider the problem of computing the quantified star size of bounded width instances. Observe that the computation of quantified star size is not strictly necessary. The algorithm of Theorem~\ref{thm:guarded} does not need to find $S$-stars for graphs of width $k$ but only for acyclic hypergraphs, which is easy by Lemma \ref{lem:edgecover}. Still it is of course desirable to know the quantified star size of an instance before applying the counting algorithm, because quantified star size has an exponential influence on the runtime. We show that for all decomposition techniques considered in this paper the quantified star size can be computed rather efficiently, roughly in $|V|^{O(k)}$ where $k$ is the width of the input. For small values of $k$, this bound is reasonable. We then proceed by showing that, on the one hand, for some decomposition measures such as treewidth or hingetree, the computation of quantified star size is even fixed parameter tractable parameterized by the width. On the other hand, we show that for decomposition measures above hypertree width it is unlikely that fixed parameter tractability can be obtained (under standard assumptions).

Instead of tackling quantified star size directly, we consider the combinatorially less complicated notion of independent sets, which is justified by the following observation:

\begin{observation}\label{obs:equivalent}
 Let $\beta$ be any decomposition technique considered in this paper. Then for every $k\in \mathbb{N}$ computing the $S$-starsize of $S$-hypergraphs of $\beta$-width at most $k$ polynomial time Turing-reduces to computing the size of a maximum independent set for hypergraphs of $\beta$-width at most $k$. Furthermore, there is a polynomial time many one reduction from computing the size of a maximum independent set in hypergraphs of $\beta$-width at most $k$ to computing the $S$-star size of hypergraphs of $\beta$-width at most $k+1$.
\end{observation}
\begin{proof}
 By definition computing $S$-starsize reduces to the computation of independent sets of $S$-components. $S$-components are induced subhypergraphs, so we get the first direction form Observation \ref{obs:subhypergraphwidth}.
 
 For the other direction let $\calH=(V,E)$ be a hypergraph for which we want to compute the size of a maximum independent set. Let $x\not\in V$. We construct the hypergraph $\calH'$ of vertex set $V'=V\cup \{x\}$ and edge set $E'=\{e\cup\{x\}\mid e\in E\}$ and set  $S:=V$. The hypergraph is one single $S$-component, because $x$ is in every edge. Furthermore, the $S$-starsize of $\calH'$ is obviously the size of a maximum independent set in $\calH$. It is easy to see that the construction increases the treewidth of the hypergraph by at most $1$ and does not increase the $\beta$-width for all other decomposition considered here at all.
\end{proof}

Because of Observation \ref{obs:equivalent} we will not talk about $S$-star size in this section anymore but instead formulate everything with independent sets.

\subsection{Exact computation}

\begin{proposition}\label{lem:indsetnk}
 There is an algorithm that given a hypergraph $\calH=(V,E)$ and a generalized hypertree decomposition $\varXi=(\calT, (\lambda_t)_{t\in T}, (\chi_t)_{t\in T})$ of $\calH$ of width $k$ computes a maximum independent set of $\calH$ in time $k |V|^{O(k)}$.
\end{proposition}
\begin{proof}
 We apply dynamic programming along the decomposition. Let $b= (\lambda, \chi)$ be a guarded block of $\calT$. Let $\calT_b$ be the subtree of $\calT$ with $b$ as its root. We set $V_b := \chi(\calT_b)$. Observe that $I\subseteq V_b$ is independent in $\calH$ if and only if it is independent in $\calH[V_b]$ so we do not differentiate between the two notions. For each independent set $\sigma\subseteq \chi$ we will compute an independent set $I_{b, \sigma} \subseteq V_b$ that is maximum under the independent sets containing exactly the vertices $\sigma$ from $\chi$. Observe that because $\lambda$ contains at most $k$ edges that cover $\chi$ we have to compute at most $k n^k$ independent sets $I_{b, \sigma}$ for each $b$.
 
 If $b$ is a leaf of $\calT$, the construction of the $I_{b, \sigma}$ is straightforward and can certainly be done in time $k|V|^{O(k)}$.
 
 Let now $b= (\lambda, \chi)$ be an inner vertex of $\calT$ with children $b_1, \ldots, b_r$. For each independent set $\sigma \subseteq \chi$ we do the following: Let $b_i= (\lambda_i, \chi_i)$, then let $\sigma_i$ be an independent set of $\chi_i$ such that $\sigma \cap \chi \cap \chi_i = \sigma_i \cap \chi \cap \chi_i$ and $|I_{b_i, \sigma_i}|$ is maximal. We claim that we can set $I_{b, \sigma} := \sigma \cup I_{b_1, \sigma_1} \cup \ldots \cup I_{b_r, \sigma_r}$.
 
 We first show that $I_{b, \sigma}$ defined this way is independent. Assume this is not true, then $I_{b, \sigma}$ contains $x,y$ that are in one common edge $e$ in $\calH[V_b]$. But then $x,y$ do not lie both in $\chi$, because $I_{b, \sigma}\cap \chi = \sigma$ and $\sigma$ is independent. By induction $x,y$ do not lie in one $V_{b_i}$ either. Assume that $x \in \chi$ and $y \in V_{b_i}$ for some $i$. Then certainly $x\notin V_{b_i}$ and $y\notin \chi$. But the edge $e$ must lie in one guard $\lambda'$ such that the corresponding block $\chi'$ contains $e$. Because of the connectivity condition for $y$ the guarded block $(\lambda', \chi')$ must lie in the subtree with root $b_i$, which contradicts $x \in e$. Finally, assume that $x\in V_{b_i}$ and $y \in V_{b_j}$ for $i\ne j$ and $x,y \notin \chi$. Then $x$ and $y$ cannot be adjacent because of the connectivity condition. This shows that $I_{b, \sigma}$ is indeed independent.
 
 Now assume that  $I_{b, \sigma}$ is not of maximum size and let $J\subseteq V_b$ be an independent set with $|J|>|I_{b, \sigma}|$ and $J\cap \chi = \sigma$. Because $J$ and $I_{b, \sigma}$ are fixed to $\sigma$ on $\chi$ there must be a $b_i$ such that $|J \cap V_{b_i}| > |I_{b_i, \sigma_i}|$. This contradicts the choice of $\sigma_i$. So $I_{b, \sigma}$ is indeed of maximum size.
 
 Because each block has at most $k |V|^k$ independent sets, all computations can be done in time $k |V|^{O(k)}$.
\end{proof}
%

\subsection{Parameterized complexity}

While the algorithm in the last section is nice in that it is a polynomial time algorithm for fixed $k$, it is somewhat unsatisfying for some decomposition techniques: If we can compute the composition quickly, we would ideally want to be able to compute the star size efficiently, too. Naturally we cannot expect a polynomial time algorithm independent of $k$, because independent set is $\NP$-complete, but we can hope for at least fixed parameter tractability with respect to $k$. We will show that this is indeed possible for some width measures, in particular tree decompositions and hingetree decompositions. On the other hand we show that this can likely not be extended to more general decomposition techniques, because independent set parameterized by hypertree width is $\W{1}$-hard.\stefan{rewrite this at some point...}

\begin{proposition}
  Given a hypergraph $\calH$ computing a maximum independent set in $\calH$ is fixed parameter tractable parameterized by the treewidth of $\calH$.
\end{proposition}

This can be seen either by applying Courcelle's Theorem of by straightforward dynamic programming. Interestingly, one can show the same result also for bounded hingetree width, which is a decomposition technique in which the bags are of unbounded size. This unbounded size makes the dynamic programming in the proof far more involved than for treewidth.

\begin{proposition}
  Given a hypergraph $\calH$ computing a maximum independent set in $\calH$ is fixed parameter tractable parameterized by the hingetree width of $\calH$.
\end{proposition}
\begin{proof}
First observe that minimum width hingetree decompositions can be computed in polynomial time \cite{GyssensJC94}, so we simply assume that a decomposition is given in the rest of the proof.

 The proof has some similarity with that of Proposition~\ref{lem:indsetnk}, so we use some notation from there. For guarded block $(\lambda, \chi)$ we will again compute maximum independent sets containing prescribed vertices. The difference is, that we can take these prescribed sets to be of size $1$: because of the hingetree condition,  only one vertex of a block may be reused in any independent set in the parent. The second idea is that we can use equivalence classes of vertices in the computations of independent sets in the considered guarded blocks, which limits the number of independent sets we have to consider. We now describe the computation in detail.
 
 Let $\varXi= (\calT, (\lambda_t)_{t\in T}, (\chi_t)_{t\in T})$ be a hingetree decomposition of $\calH$ of width $k$. Let $b=(\lambda, \chi)$ be a guarded block of~$\varXi$ and let $b'=(\lambda', \chi')$ be its parent. As before, let  $\calT_b$ be the subtree of $\calT$ with $b$ as its root and $V_b := \chi(\calT_b)$. Set $\calH_b := (V_b, E_b)$ with $E_b := \bigcup \lambda^*$ with the union being over all guarded blocks in $\calT_b$. The main idea is to iteratively compute, for all vertices  $v\in \chi' \cap \chi$, a maximum independent set $J_{v,b}$ in $\calH_b = (V_b, E_b)$ containing $v$. Furthermore, we also compute an independent set $J_{\emptyset,b}$ that contains no vertices of $\chi' \cap \chi$.   
 Note that, since $\chi\subseteq \bigcup_{e\in \lambda} e$, there are no isolated vertices in $\chi$ and the size of a maximum independent set is bounded by $k$ in each block.
   
 For a node $b=(\lambda, \chi)$, we organize the vertices in $\chi$ into at most $2^k$ equivalence classes by defining $v$ and $u$ to be equivalent if they lie in the same subset of edges of $\lambda$. The equivalence class of $v$ is denoted by $\classe{v}$. For each class, a representant is fixed. We denote by $\bar v$, the representant of the equivalence class of $v$ and by $\bar \chi\subseteq \chi$, the restriction of $\chi$ on these at most $2^k$ representants.
 
 Let first $b$ be a leaf. We  first compute independent sets on $\bar \chi$. Observe that the independent sets are invariant under the choice of representants. For each equivalence class  $\classe{v}$, we compute  $J_{\bar v,b}\subseteq \bar\chi$ as a maximum independent set containing $\bar v$. Computing the classes and a choice of  maximum independent sets containing each $\bar v$ can be done in time $k 2^{k^2}$ because independent sets cannot be bigger than $k$. Clearly, $J_{v,b}$, a maximum independent set containing $v$, can be easily computed from the set $J_{\bar v,b}$. Thus, one can compute all the $J_{v,b}$ in time $k 2^{k^2}n$. 
 The computation of $J_{\emptyset, b}$ can be done on representants, too, by simply excluding the vertices from $\chi' \cap \chi$. 
 
 Let $b$ now be an inner vertex 
 and $b_1, b_2,..., b_m$ be its children with $b_i=(\lambda_i, \chi_i)$, $i\in [m]$. We again consider equivalence classes on $\chi$. Fix $v\in \chi$ and compute the list $L_{\bar v,b}$ of all independent sets $\sigma\subseteq \bar \chi$ containing $\bar v$. Fix now $\sigma\in L_{\bar v,b}$. We first compute a set $J_{v,b}^{\sigma}$ as a maximum independent set of $\calH_b$ containing $v$ and whose  vertices in $\chi$  have the representants  $\sigma$. We will distinguish for a given vertex $\bar u\in \sigma$ if it is the representant of a vertex belonging to the block of some (or several) children of $b$ or if it represents vertices of $\chi\backslash (\bigcup_{i= 1}^m \chi_i)$  only. Therefore we partition $\sigma$ into  $\sigma', \sigma''$ accordingly:
 \begin{itemize}
 \item $\sigma := \sigma' \cup \sigma''$
 \item $\sigma' := \bar\chi \cap \{\bar u \mid u \in \bigcup_{i= 1}^m \chi_i\}$.
 \item $\sigma'':= \bar{\chi} \backslash \{\bar u \mid u \in \bigcup_{i= 1}^m \chi_i\}$
 \end{itemize}
 
  Set $\sigma':=\{\bar u_1,...,\bar u_h\}$ 
  with $h\leq m$.
 Let us examine the consequences of $\calT$ being a hingetree decomposition. We have that, for all $i\in [m]$, there exists $e_i\in \lambda$, such that $\chi\cap \chi_i \subseteq e_i$. Thus, since $\sigma$ is an independent set in $\bar{\chi}\subseteq \chi$, at most one vertex in $\sigma'$ is a representant of a vertex in $\chi_i$. Thus
 
 \begin{equation}\label{equ:distinctclasses}
 \forall u\neq u' \in \sigma \colon \chi_i\cap \classe{u}= \emptyset \vee \chi_i\cap \classe{u'}= \emptyset .
 \end{equation} 
 
We denote by $S_i=\{j\mid  \classe{u_i}\cap\chi_j \neq \emptyset\}$ and by $S=[m]\backslash \bigcup S_i$. By (\ref{equ:distinctclasses}) the sets $S_1,...,S_h,S$ form a partition of~$[m]$. To construct $J_{v,b}^{\sigma}$, we now determine for each $i\leq h$, which vertex $u$ of $\classe{u_i}$ can contribute the most, by taking the union of all the maximum independent sets $J_{u,b_j}$, $j\in S_i$, it induces at the level of the children of $b$.  

For each fixed $u\in \classe{u_i}$,
let 
\[I_{i,u}=  \{u\}\cup \bigcup_{j\in S_i} J_{u,b_j}, \]
where we set $J_{u,b_j}:=J_{\emptyset,b_j}$ if $u\notin \chi_j$.
Let then $I_i=I_{i,u}$ for some $u\in \classe{u_i}$ for which the size of $I_{i,u}$ is maximal. 
 
 The set  $J_{v,b}^{\sigma}$ is now obtained as follows
  depending on whether  $\bar v\in \sigma''$ or $\bar v\in \sigma'$.
 If  $\bar v\in \sigma''$,  we claim that $J_{v,b}^{\sigma}$ can be chosen as
  \[ J_{v,b}^{\sigma}:=\{v\} \cup (\sigma''\backslash \{\bar v\} ) \cup \bigcup_{i=1}^h I_i \cup \bigcup_{i\in S} J_{\emptyset,b_i}.\]
 
 If $\bar v\in \sigma'$, say  $\bar v=u_1$,  we  claim that $J_{v,b}^{\sigma}$ can be chosen as
  \[ J_{v,b}^{\sigma}:=\sigma''  \cup\bigcup_{j\in S_1: v\in \chi_j} J_{v,b_j}\cup\bigcup_{j\in S_1: v\notin \chi_j} J_{\emptyset,b_j}  \cup \bigcup_{i=2}^h I_i \cup \bigcup_{i\in S} J_{\emptyset,b_i}.\]
 
 The set $J_{v,b}$ is taken as one of the sets  $J_{v,b}^{\sigma}$ of maximal size for a $\sigma\in L_{v,b}$.
 To compute $J_{\emptyset, b}$, the arguments are similar.

 We first show that all $J_{v, b}$ are indeed independent sets in~$\calH_b$. Clearly, it is enough to prove this for any  $J_{v,b}^{\sigma}$. There will be no reason to distinguish whether  $\bar v\in \sigma''$ or $\bar v\in \sigma'$, because our arguments will apply to all $J_{v, b}^{\sigma}$ independent of the choice of a distinguished element $v$.
 We will make extensive use of the two following facts.

 \begin{itemize}
 \item Let $j,j'\in [m]$ and $I\subseteq V_{b_j},I'\subseteq V_{b_{j'}}$ independent sets of $\calH_{b_j}$ and $\calH_{b_j'}$ respectively. By the connectivity condition for tree decomposition we have
 \[ I\cap I' \subseteq \chi_j\cap \chi_{j'}\cap \chi. \]
 This permits to investigate the intersection of two independent sets $I, I'$ by looking at their restriction on $\chi$.
 
 \item Let now $I\subseteq V_{b_j}$ be an independent set of $\calH_{b_j}$. Then,  $I$ remains an independent set in $\calH_b$. Indeed, suppose there is a $e\in E_b\backslash E_{b_j}$ containing two vertices $y_1,y_2\in I$. Since all edges must belong to a guard, there exists a node $b^*=(\lambda^*, \chi^*)$ such that $e\in \lambda^*$. Then, since in a hingetree decomposition we have $\chi^*=\bigcup \lambda^*$, then $\{y_1,y_2\}\subseteq  e\subseteq \chi^*$. But then, by the connectivity condition it follows that $\{y_1,y_2\}\subseteq  \chi$.  Hence, by the intersection property of hingetree decomposition, there exists $e_j\in \chi_j$ such that 
  \[ \{y_1,y_2\}\subseteq \chi \cap \chi_j \cap e_j \]
\noindent which implies that $y_1$ and $y_2$ are adjacent in $\calH_{b_j}$. Contradiction. 
\end{itemize}
 
 We now start the proof that $J_{v,b}^{\sigma}$ is independent incrementally.
 Let $i\in [h]$, $u\in \classe{u_i}$ and $j\in S_i$ and consider the set $I:=J_{u,b_j}$. By induction, the set $I$ is independent in $\calH_{b_j}$. By the hingetree condition, there exists $e_j\in\lambda_j$ such that $\chi\cap \chi_j\subseteq e_j$. By the connectivity condition, this implies 
 $\chi\cap I\subseteq e_j$. Then, since $I$ is an independent set, no two vertices of $\chi$ can belong to $I$ i.e.\ $|\chi\cap I|\leq 1$. 
 The connectivity condition also implies that, for $j'\neq j$, $V_{b_{j'}}\cap I\subseteq \chi\cap \chi_j$, hence  $|V_{b_{j'}}\cap I|\leq 1$ and $I$ is an independent set of $\calH_b$.
 Finally, the set $I_i=\bigcup_{j\in S_i} J_{u,b_j}$ is also an independent set of $\calH_b$, since  for any distinct $j,j'\in S_i$:
  \[ J_{u,b_j}\cap J_{u,b_{j'}} \subseteq \chi_j\cap \chi_{j'} \cap \chi \subseteq e_j.\] 
 
 Hence $ J_{u,b_j}\cap J_{u,b_{j'}}$  contains at most one vertex (which is in $\chi$ and could then only be $u$).
  
  Let now $i,i'\in [m]$ be distinct. By the arguments above, $I_i$ (resp. $I_{i'}$) contains at most one element $u$ (resp. $u'$) such that $u\in \classe{u_i}$ (resp. $u'\in \classe{u_{i'}}$). By Equation~\ref{equ:distinctclasses}, we have that the two classes are distinct and that $u_i\neq u_{i'}$. But $u_i,u_{i'}\in \sigma$ and $\sigma$ is independent in $\chi$. Hence, $u_i,u_{i'}$ cannot be adjacent in $\calH_b$. 
 Consequently,
   \[   \bigcup_{i=1}^h I_i \]
 \noindent is an independent set in $\calH_b$.

 Let $j\in S$. $J_{\emptyset,b_j}$ is independent in $\calH_{b_j}$ and $J_{\emptyset,b_j}\subseteq V_{b_j}\backslash \chi$. Hence, $J_{\emptyset,b_j}$ is independent in $\calH_{b}$. This also implies that, given $j'\in [m]$ distinct from $j$, 
 $J_{\emptyset,b_j}\cap V_{b_{j'}}=\emptyset$. Thus,
 \[ \bigcup_{i=1}^h I_i \cup \bigcup_{i\in S} J_{\emptyset,b_i}.\]
 \noindent is independent in $\calH_b$.
 
 Finally, by construction, for all $i\in [h]$, $I_i\cap \chi=\{u\}$  with  $\bar u=\bar u_i\in \sigma'$. 
 Also $\sigma=\sigma'\cup \sigma''$ is independent in $\chi$ hence in $\calH_b$.  
 No vertices  $y_1\in I_i$ and $y_2\in \sigma''$ can be adjacent because, again, this would imply that $\{y_1,y_2\}\subseteq \chi$ and contradict the fact that $\bar y_1,\bar y_2$ are independent in $\sigma$. Thus
 $ J_{v,b}^{\sigma}$ is independent.
 
\medskip

 We now prove that $J_{v,b}$ is of maximum size. Observe that it suffices to show this again for each $J_{v,b}^\sigma$. Each maximum independent set $J$ of $\calH_b$ that contains $v$ and whose vertices in $\chi$ have exactly the representants $\sigma$ can be expressed as $\tau\cup J_1 \cup J_2 \cup ... \cup J_m$. Here $\tau\subseteq \chi$ is an independent set of $b$ containing $v$ and whose representants are $\sigma$. Furthermore, $J_i$ is an independent set of $\calH_b$ that contains only vertices of $V_{b_i}$. The set $J_i$ may only contain one vertex $u_i$ from $\chi\cap \chi_i$. But then exchanging $J_i$ for $J_{u_i, b_i}$ may only increase the size of the independent set, so we can assume that $I$ has the form  $\tau\cup J_{u_1, b_i} \cup J_{u_2, b_2} \cup \ldots \cup J_{u_m, b_m}$ where $u_i$ may also stand for~$\emptyset$.
 
 Assume now that $J_{v,b}^\sigma$ is not maximum, i.e.\ there is an independent set $J$ containing $v$ whose vertices in $\chi$ have the representants $\sigma$ and $J$ is bigger than $J_{v,b}^\sigma$. Then one of four following things must happen:
 \begin{itemize}
  \item There is an $i$ such that $v\in \chi_i$ and $J\cap V_{b_i}$ is bigger than $J_{v,b_i}$. But this case cannot occur by induction.
  \item $v=u_1$ and there is a $j\in S_1$ such that $v\notin \chi_j$ and $|J\cap V_{b_j}|> |J_{\emptyset,b_j}|$. By induction we know that $J_{\emptyset,b_j}$ is optimal under all independent sets of $\calH_{b_j}$ not containing any vertex of $\chi_j\cap \chi$, so there must be a vertex $u\in J\cap \chi\cap \chi_j$. Since $J$ is independent,  $v$ and $u$ share no edge in $\lambda$ and then $\bar v \ne \bar u$. Since $j\in S_1$, it holds that $\classe{v}\cap \chi_j\neq \emptyset$ and by Equation~\ref{equ:distinctclasses}, $\classe{u}\cap \chi_j= \emptyset$. Contradiction.

  \item There is an $i\in S$ such that $J\cap V_{b_i}$ is bigger than $J_{\emptyset, b_i}$. But from $i\in S$ it follows by definition that $\chi \cap \chi_i \cap J= \emptyset$, so this case can not occur by induction, either.
  
  \item There is an $i\in [h]$ 
   such that $|J\cap (\bigcup_{j\in S_i} V_j)|> |I_i|$. We claim that $(\bigcup_{j\in S_i} \chi_j) \cap \chi\cap J$ contains only one vertex. Assume there are two such vertices $x$ and $y$. 
   By definition, $\bar x, \bar y\in \bar \tau$. Since $J$ is independent, $\bar x$ and $\bar y$ are not adjacent in $\bar\chi$ and $\bar x \ne \bar y$. 
   At least one of these, say $y$, must be in $\classe{u_i}$, because $\bar u_i \in \bar \tau$ by definition.  Let $x\in V_{j'}$ with $j'\in S_i$, then there is a vertex $w\in \classe{u_i} = \classe{y}$ in $\chi_{j'}\cap \chi \subseteq e_j$ by definition of $S_i$. But then $\bar x$ and $\bar y$ are adjacent in $\bar \chi$ which is a contradiction.
  
  So there is exactly one vertex $u$ in $(\bigcup_{j\in S_i} \chi_j) \cap \chi\cap J$. But then $|J\cap (\bigcup_{j\in S_i} V_j)| > I_{i,u}$. Thus either there must be a $j\in S_i$ with $u\in V_j$ such that $|J\cap V_j|> |J_{u, b_j}|$ or there must be a $j\in S_i$ with $u\notin V_j$ such that $|J\cap V_j|> |J_{\emptyset, b_j}|$. The former clearly contradicts the optimality of $J_{u, b_j}$, while the latter leads to a contradiction completely analogously to the second item above.
 \end{itemize}
 
 Because only $ k 2^{k^2} n^2$ sets have to be considered for each guarded block, this results in an algorithm with runtime $ k 2^{k^2} |V|^{O(1)}$.
\end{proof}

\begin{lemma}\label{lem:starsizehard}
 Computing maximum independent sets on hypergraphs is $\W{1}$-hard parameterized by generalized hypertree width.
\end{lemma}
\begin{proof}
 We will show a reduction from $p$-$\mathrm{IndependentSet}$ which is the following problem: Given a graph $G$ and an integer $k$ which is the parameter, decide if $G$ has an independent set of size $k$. Because $p$-$\mathrm{IndependentSet}$ is well known to be $\W{1}$-hard, this suffices to establish $\W{1}$-hardness of independent set parameterized by hypertree width.
 
 So let $G=(V,E)$ be a graph and let $k$ be a positive integer. We construct a hypergraph $\calH=(V', E')$ in the following way: For each vertex $v$ the hypergraph $\calH$ has $k$ vertices $v_1, \ldots, v_k$. For $i=1,\ldots, k$ we have an edge $V_i:= \{v_i \mid v\in V\}$ in $E'$. Furthermore, for each $v\in V$  we add an edge $H_v:=\{v_i \mid i\in [k]\}$. Finally we add the edge sets $E_{ij}:=\{v_i u_j\mid uv\in E\}$ for $i,j\in [k]$. $\calH$ has no other vertices or edges. The construction is illustrated in Figure \ref{fig:hardness}.
 
 \begin{figure*}[t]
 \begin{center}
 \includegraphics[scale=0.60]{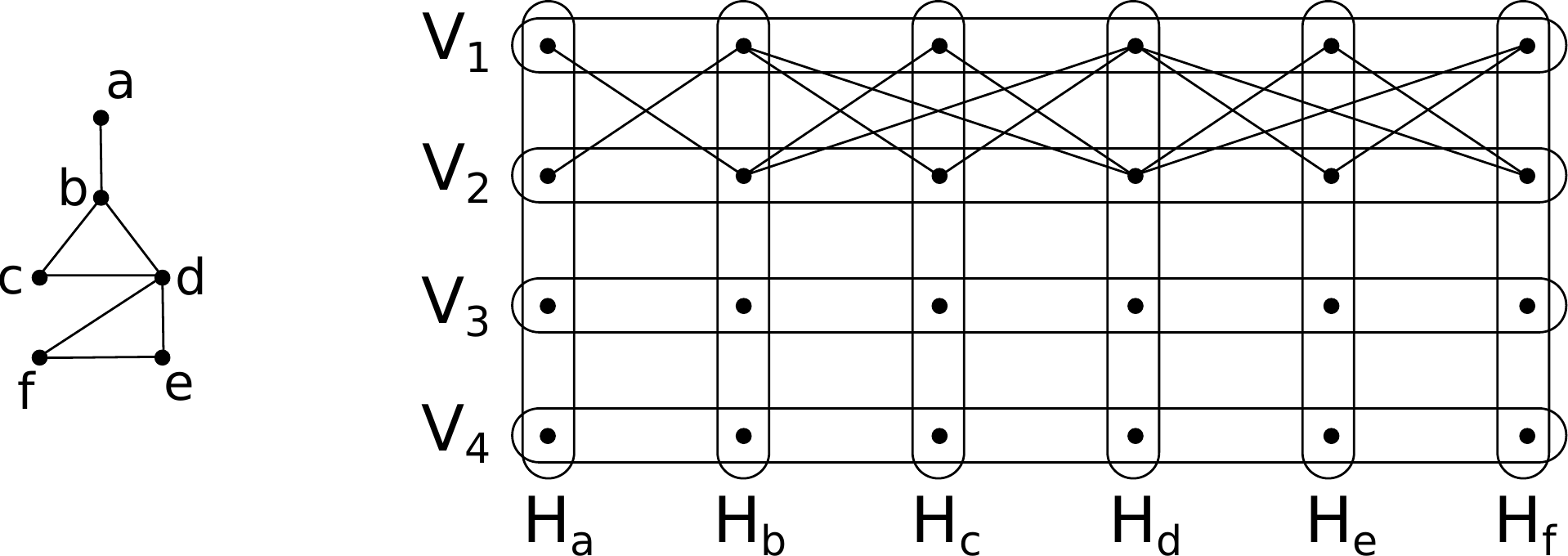}
 \caption{We illustrate the construction for Lemma \ref{lem:starsizehard} by an example. A graph $G$ on the left with the associated hypergraph $\calH$ for $k=4$ on the right. To keep the illustration more transparent the edge sets $E_{ij}$ are not shown except for $E_{1,2}$ and $E_{2,1}$.}\label{fig:hardness}
 \end{center}
 \end{figure*}
 
 We claim that $G$ has an independent set of size $k$ if and only if $\calH$ has an independent set of size $k$. Indeed, if $G$ has an independent set $v^1,\ldots , v^k$, then $v_1^1, \ldots v_k^k$ is easily seen to be an independent set of size $k$ in $\calH$. Now assume that $\calH$ has an independent set $I$ of size $k$. Then for each $v\in I$ we can choose a vertex $\pi(v)\in V$ such that $v\in H_{\pi(v)}$. Furthermore for distinct $v,u\in I$ the corresponding vertices $\pi(v),\pi(u)$ have to be distinct, too, so $\pi(I)\subseteq V$ has size $k$. Finally, we claim that $\pi(I)$ is independent in $G$. Assume this is not true, then there are vertices $\pi(v), \pi(u)$ such that $\pi(v) \pi(u)\in E$. But then $vu\in E'$ by construction which is a contradiction. So, indeed $G$ has an independent set of size $k$ if and only if $\calH$ has one.
 
 We now show that $\calH$ has generalized hypertree width at most $k$ by constructing a generalized hypertree decomposition $(\calT, (\lambda_t)_{t\in T}, (\chi_t)_{t\in T})$ of $\calH$ of width $k$. The tree $\calT$ only consists of one single vertex $v$, the block of $v$ is $\chi_v:= V'$ and the guard is $\lambda_t := \{V_1, \ldots , V_k\}$. It is easily seen that this is indeed a hypertree decomposition of width $k$.
 
 Observing that the construction of $\calH$ from $G$ can be done in time polynomial in $|V|$ and $k$ completes the proof.
\end{proof}

\subsection{Approximation}

We have seen that computing maximum independent sets of hypergraphs with decompositions of width $k$ can be done in polynomial time for fixed width $k$ and that for some decompositions it is even fixed parameter tractable with respect to $k$. Still, the exponential influence of $k$ is troubling. In this section we will show that we can get rid of it if we are willing to sacrifice the optimality of the solution. We give a $k$-approximation algorithm for computing  maximum independent sets of graphs with generalized hypertree width~$k$ assuming that a decomposition is given. We start by formulating a lemma.

\begin{lemma}\label{lem:approx}
Let $\calH$ be a hypergraph with a generalized hypertree decomposition $\varXi = (\calT, (\lambda_t)_{t\in T}, (\chi_t)_{t\in T})$ of width~$k$. Let $\calH'=(V,E')$ where $E':= \{\chi_t\mid t\in T\}$. Let $\ell$ be the size of a maximum independent set in $\calH$ and let $\ell'$ be the size of a maximum independent set in $\calH'$. Then
\[\frac{\ell}{k}\le \ell' \le \ell.\]
\end{lemma}

Before we prove Lemma \ref{lem:approx} we will show how to get the approximation algorithm from it.

\begin{observation}\label{obs:indset}
 Every independent set of $\calH'$ is also an independent set of $\calH$.
\end{observation}
\begin{proof}
 Each pair of independent vertices $x,y$ in $\calH'$ is by definition only in different blocks $\chi_t$ in $\calH$. For each edge $e\in E$ there must by definition of generalized hypertree decompositions be a block $\chi$ such than $e\subseteq \chi$. Thus no edge $e\in E$ can contain both $x$ and $y$, so $x$ and $y$ are independent in $\calH$, too. 
\end{proof}

\begin{corollary}
 There is a polynomial time algorithm that given a hypergraph $\calH$ and a generalized hypertree decomposition of width $k$ computes an independent set of size~$\ell$ of $\calH$ such that $|I|\ge \frac{\ell}{k}$ where $\ell$ is the size of a maximum independent set of $\calH$. 
\end{corollary}
\begin{proof}
Observe that $\calH'$ is acyclic. With Lemma \ref{lem:edgecover} we compute a maximum independent set $I$ of $\calH'$ whose size by Lemma \ref{lem:approx} only differs by a factor $\frac{1}{k}$ from $\ell$. By Observation~\ref{obs:indset} we have that $I$ is also an independent set of~$\calH$.
\end{proof}

\begin{proofof}{Lemma \ref{lem:approx}}
The second inequality follows directly from Observation \ref{obs:indset}.
 
For the first inequality consider a maximum independent set $I$ of $\calH$. Observe that a set $I'$ is an independent set of $\calH'$ if and only if it is an independent set of its primal graph~$\calH'_P$, so it suffices to show the same result for $\calH'_P$.
 
\begin{clm}\label{clm:treewidht}
$\calH'_P[I]$ has treewidth at most $k-1$.
\end{clm}
\begin{proof}
 We construct a tree decomposition from $\varXi$. To do so consider $\varXi[I]$ which for each guarded block $(\lambda, \chi)$ of~$\varXi$ contains $(\lambda', \chi')$ where $\lambda' := \{ e\cap I \mid e\in \lambda, e\cap I \ne \emptyset \}$ and $\chi':= \chi \cap I$. The set $I$ is independent, so each guard of $\varXi'[I]$ is a set of singletons and if follows $|\chi'|\leq |\lambda'|$ for each guarded block $(\lambda', \chi')$.

Let $\calT[I]$ be the tree of $\varXi[I]$ induced by $\calT$ in the obvious way. Then the blocks $\chi' = \chi\cap I$ still fulfill the connectedness condition. Furthermore, for each edge $uv$ in $\calH'[I]$ there is a guarded block $(\lambda', \chi')$ such that $u,v \in \chi'$. Thus $\varXi[I]$ is a tree decomposition of $\calH'_P[I]$. But we have that $|\chi'| \le |\lambda'| \le |\lambda| \le k$ and thus the tree decomposition is of width at most $k-1$.
\end{proof}

\begin{clm}
 $\calH'_P[I]$ has an independent set $I'$ of size at least~$\frac{|I|}{k}$.
\end{clm}
\begin{proof}
From Claim \ref{clm:treewidht} it follows that $\calH'[I]$ and all of its subgraphs have a vertex of degree at most $k$ (see e.g.\ \cite[p.~265]{FlumGrohe06}). We construct $I'$ iteratively by choosing a vertex of minimum degree and deleting it and its neighbors from the graph. In each round we delete at most $k$ vertices, so we can choose a vertex in at least $\frac{|I|}{k}$ rounds. Obviously the chosen vertices are independent.
\end{proof}

Every independent set of $\calH_P[I]$ is also an independent set of $\calH_P$ which completes the proof of Lemma \ref{lem:approx}.
\end{proofof}

\section{Fractional Hypertree width}\label{sct:fractionalstatements}

In this section we extend the main results of the paper to fractional hypertree width, which is the most general  notion known so far that leads to tractable Boolean $\CQ$ \cite{GroheMarx06}. In particular it is strictly more general than generalized hypertree width.

\begin{definition}
 Let $\calH = (V,E)$ be a hypergraph. A \emph{fractional edge cover} of a vertex set $S\subseteq V$ is a mapping $\psi: E\rightarrow [0,1]$ such that for every $v\in S$ we have $\sum_{e\in E: v\in e} \psi(e) \ge 1$. The weight of $\psi$ is $\sum_{e\in E} \psi(e)$. The fractional edge cover number of $S$, denoted by $\rho_\calH^*(S)$ is the minimum weight taken over all fractional edge covers of $S$.
 
 A \emph{fractional hypertree decomposition} of $\calH$ is a triple $(\calT, (\chi_t)_{t\in T}, (\psi_t)_{t\in T})$ where $\calT = (T,F)$ is a tree, and $\chi_t \subseteq V$ and $\psi_t$ is a fractional edge cover of $\chi_t$ for every $t\in T$  satisfying the following properties:
 \begin{enumerate}
  \item For every $v\in V$ the set $\{t\in T \mid v\in \chi_t\}$ induces a subtree of $\calT$.
  \item For every $e\in E$ there is a $t\in T$ such that $e\subseteq \chi_t$.
 \end{enumerate}
The width of a fractional hypertree decomposition $(\calT, (\chi_t)_{t\in T},$ $ (\psi_t)_{t_in T})$ is $\max_{t\in T}(\rho_\calH^*(\chi_t))$. The \emph{fractional hypertree width} of $\calH$ is the minimum width over all fractional hypertree decompositions of $\calH$.
\end{definition}

Together with the previous results of this paper, the two following ones will serve as key ingredients to prove the main results of this section.

\begin{theorem}[\cite{GroheMarx06}]\label{thm:enumerate}
 The solutions of a $\CQ$ instance $\Phi$ with hypergraph $\calH$ can be enumerated in time $\|\Phi\|^{\rho^*(\calH) + O(1)}$.
\end{theorem}

\begin{theorem}[\cite{Marx10}]\label{thm:approxFHW} Given a hypergraph $\calH$ and a rational number $w\ge 1$, it is possible in time $\|\calH\|^{O(w^3)}$ to either
 \begin{itemize}
  \item compute a fractional hypertree decomposition of $\calH$ with width at mots $7w^3+31w+7$, or
  \item correctly conclude that $\fhw(\calH)\ge w$.
 \end{itemize}
\end{theorem}

\subsection{Tractable counting}

We start of with the quantifier free case which we will use as a building block for the more general result later.

\begin{lemma}\label{lem:fractQF}
 The solutions of a quantifier free $\CQ$ instance $\Phi$ with hypergraph $\calH$ can be counted in time $\|\Phi\|^{\fhw(\calH)^{O(1)}}$.
\end{lemma}
\begin{proof}
 With Theorem \ref{thm:approxFHW} we can compute a fractional hypertree decomposition $(\calT, (B_t)_{t\in T}, (\psi_t)_{t\in T})$ of width at most $k:=O(\fhw(\calH)^3)$. For each bag $B_t$ we can with Theorem \ref{thm:enumerate} in time $\|\Phi\|^k$ compute all solutions to the CQ $\Phi[B_t]$ that is induced by the variables in $B_t$. Let these solutions form a new relation $\calR_t$ belonging to a new atom $\varphi_t$. Then $\bigwedge_{t\in T} \varphi_t(B_t)$ gives a solution equivalent, acyclic, quantifier free $\sCQ$ instance of size~$\|\Phi\|^{O(k)}$.
\end{proof}

We can now formulate a version of Theorem \ref{thm:guarded} for fractional hypertree width.

\begin{theorem}\label{thm:fractionalstatement}
 There is an algorithm that given a $\sCQ$-instance $\Phi$ of quantified starsize $\ell$ and fractional hypertree width $k$ counts the solutions of $\Phi$ in time $\|\Phi\|^{p(k,\ell)}$ for a polynomial $p$.
\end{theorem}

\begin{proof} This is a minor modification of the proof of Theorem \ref{thm:guarded}. 
Let $\calH= (V,E)$ be the hypergraph of $\Phi$. Because of Theorem \ref{thm:approxFHW} we may assume that we have a fractional hypertree decomposition $\varXi:=(\calT, (\chi_t)_{t\in T}, (\psi_t)_{t\in T})$ of width $k':= k^{O(1)}$ of $\calH$ where $\calH$ is the hypergraph of $\Phi$. 
For each edge $e\in E$ we let $\varphi(e)$ be the atom of $\Phi$ that induces $e$. 

Let $V_1, \ldots, V_m$ be the vertex sets of the components of $\calH-S$ and let $V_1', \ldots, V_m'$ be the vertex sets of the $S$-components of $\calH$. Clearly, $V_i\subseteq V_i'$ and $V_i'-V_i = V_i'\cap S =: S_i$. Let $\Phi_i$ be the restriction of $\Phi$ to the variables in $V_i'$ and Let $\varXi_i$ be the corresponding fractional hypertree decomposition. Then $\varXi_i$ has a tree $\calT_i$ that is a subtree of $\calT$. 

For each $\Phi_i$ we construct  a new $\sCQ$-instance $\Phi_i'$ by computing for each bag $B_t$ a constraint $\varphi$ in the variables $B_t$ that contains the solutions of $\Phi_i[B_t]$. The decomposition $\varXi$ has width at most $k'$ so this can be done in time $n^{O(k')}$ by Theorem \ref{thm:enumerate}. Obviously $\Phi_i$ and $\Phi'_i$ are solution equivalent and $\Phi'_i$ is acyclic. Furthermore, $\Phi_i'$ has only one single $S_i$-component, because all the vertices in $V_i$ are connected in $\Phi$ and thus also in $\Phi_i'$. Let $\calH_i$ be the hypergraph of $\Phi_i'$, then $\calH_i$ has $S_i$-star size at most $\ell$. Thus the vertices in $S_i$ can be covered by at most $\ell$ edges in $\calH_i$ by Lemma \ref{lem:edgecover}. Pick $\ell$ such edges $e_1, \ldots , e_\ell$. We construct a new atom $\varphi_i$ in the variables $S_i$ that is solution equivalent to $\Phi_i'$ by doing the following: For each combination $t_1, \ldots, t_\ell$ of tuples in $\varphi(e_1), \ldots, \varphi(e_\ell)$ fix the free variables in $\Phi_i'$ to the constants prescribed by the tuples $t_1, \ldots, t_\ell$ if these do not contradict. If the resulting $\ACQ$ instance has a solution, add $t_1\bowtie \ldots \bowtie t_k$ to the relation of $\phi_i$.

We now eliminate all quantified variables in $\Phi$. To do so we add the constraint $\phi_i$ for $i\in [m]$ and delete all constraints that contain any quantified variable, i.e.\ we delete each $\Phi_i'$. Call the resulting $\sCQ$ instance $\Phi'$. Because $\varphi_i$ is solution equivalent to $\Phi_i'$, we have that $\Phi$ and $\Phi'$ are solution equivalent, too. 

We then construct a fractional hypertree decomposition of $\Phi'$ by doing the following: we set $B_t' = (B_t\setminus \bigcup_{i\in I_t}V_i)\cup \bigcup_{i\in I_t} S_i$ for each bag $B_t$ where $I_t:= \{i \mid B_t\cap V_i \ne 0\}$. For each bag $B_t$ we construct a fractional edge cover $\psi_t'$ of $B_t'$ by setting $\psi_t'(e) := \psi_t(e)$ for all old edges and setting $\psi_t(S_i) = 1$ for $i\in  I_t$ where $S_i$ corresponds to the newly added constraint $\phi_i$ with $B_t\cap V_i \ne 0$. The result is indeed a fractional edge cover of width at most $k'$, because each variable not in any $S_i$ is still covered as before and the variables in $S_i$ are covered by definition of $\psi_t$. Furthermore, we claim that the width of the cover is at most $k'$. Indeed, for each $i\in I$ we had for each $v\in V_i$ $\sum_{e\in E: v\in e} \psi(e) \ge 1$. None of these edges appears in the new decomposition anymore. Thus adding the edge $S_i$ with weight $1$ does not increase the total weight of the cover. It is now easy to see that doing this construction for all $B_t$ leads to a fractional hypertree decomposition of $\Phi'$ of width at most $k'$.

Applying Lemma \ref{lem:fractQF} concludes the proof.
\end{proof}

\subsection{Computing independents sets}
Also $S$-star size or equivalently independent sets of bounded fractional hypertree width hypergraphs can be computed efficiently.

\begin{lemma}\label{lem:enumerateInd}
 The independent sets of a hypergraph $\calH=(V,E)$ can be enumerated in time $|\calH|^{O(\rho^*_\calH(V))}$.
\end{lemma}
\begin{proof}
 Let $\calH=(V,E)$. We construct a conjunctive query $\Phi$ with the hypergraph $\calH$. Let $V$ be the variables of $\Phi$, $\{0,1\}$ the domain and add a relation $\calR_e$ for each $e\in E$. The relation $\calR_e$ has all tuples that contain at most one $1$ entry. Finally, $\Phi$ has the formula $\bigwedge_{e\in E} \calR_e(e)$.
 
 Clearly, $\Phi$ has indeed the hypergraph $\calH$. Furthermore the solutions of $\Phi$ are exactly the characteristic vectors of independent sets of $\Phi$. Thus we can enumerate all independent sets of $\calH$ in time $|\calH|^{O(\rho^*)}$ with Theorem \ref{thm:enumerate}.
\end{proof}

\begin{lemma}\label{lem:fractionalISstatement}
 There is an algorithm that given a hypergraph $\calH=(V,E)$ of fractional hypertree width at most  $k$  computes a maximum independent set of $\calH$ in time $|\calH|^{k^{O(1)}}$.
\end{lemma}

\begin{proof}
 Dynamic programming along a fractional hypertree decomposition. 
  In a first step we compute a fractional hypertree decomposition $(\calT, (B_t)_{t\in T}, (\psi_t)_{t\in T})$ of width $k' = k^{O(1)}$ of $\calH$ with Theorem \ref{thm:approxFHW}. For each bag $B_t$ we then compute  all independent sets of $\calH[B]$ with Lemma \ref{lem:enumerateInd}, call this set $I_t$.
 
 By dynamic programming similar to the proof of Lemma~\ref{lem:indsetnk} we then compute a maximum independent set of $\calH$.\stefan{OK, we should provide more details at some point, but I am very optimistic this works.}
\end{proof}

\section{Conclusion}

The results of this paper give a clear picture of tractability for counting solutions of conjunctive queries for structural classes that are known to have tractable decision problems. Essentially counting is tractable if and only if these classes are combined with quantified star size. So to find more general structural classes that allow tractable counting, progress for the corresponding decision question appears to be necessary.

Another way of generalizing the results of this paper would be extending the logic that the queries can be formulated in. Just recently Chen and Dalmau \cite{ChenD12} have characterized the tractable classes of bounded arity $\mathrm{QCSP}$ which is essentially a version of $\CQ$ in which also universal quantifiers are allowed. They do this by introducing a new width measure for first order $\{\forall, \exists, \land\}$-formulas. We conjecture that their width measure also characterizes the tractable cases for $\mathrm{\#QCSP}$, i.e.\ tractable decision and counting coincide here. It would be interesting to see how far this can be pushed for the case of unbounded arity. 

Another extension appears in a recent paper by Chen~\cite{Chen12} where he considers existential formulas that may use conjunction and disjunction. This is particularly interesting, because it corresponds to the classical select-project-join queries with union that play an important role in database theory. One may wonder if Chen's results may be extended to counting, too.


\paragraph*{Acknowledgements} The authors are grateful for the very helpful feedback on this paper they got 
from the reviewers of the conference version.

\bibliographystyle{abbrv}
\bibliography{bibfiledecomp}

\end{document}